\documentclass[11pt]{amsart}
\usepackage{amssymb,mathrsfs,graphicx,enumerate}
\usepackage{amsmath,amsfonts,amssymb,amscd,amsthm,bbm}
\allowdisplaybreaks
\usepackage[retainorgcmds]{IEEEtrantools}
\usepackage{colortbl}
\usepackage{kotex}
\usepackage{graphicx, subfigure}
\allowdisplaybreaks
\title[The Lohe Hermitian sphere model with time delay]{Emergent behaviors of homogeneous Lohe Hermitian sphere particles under time-delayed interactions}

\author[Ha]{Seung-Yeal Ha}
\address[Seung-Yeal Ha]{\newline Department of Mathematical Sciences and Research Institute of Mathematics \newline Seoul National University, Seoul 08826 and \newline
Korea Institute for Advanced Study, Hoegiro 85, Seoul, 02455, Republic of Korea}
\email{syha@snu.ac.kr}
 
\author[Hwang]{Gyuyoung Hwang}
\address[Gyuyoung Hwang]{\newline Department of Mathematical Sciences\newline Seoul National University, Seoul 08826, Republic of Korea}
\email{hgy0407@snu.ac.kr}

\author[Park]{Hansol Park}
\address[Hansol Park]{\newline Department of Mathematical Sciences\newline Seoul National University, Seoul 08826, Republic of Korea}
\email{hansol960612@snu.ac.kr}

\newtheorem{theorem}{Theorem}[section]
\newtheorem{lemma}{Lemma}[section]

\newtheorem{proposition}{Proposition}[section]

\newtheorem{definition}{Definition}[section]

\newcommand{\bbr}{\mathbb R}

\newcommand{\bbs}{\mathbb S}

\newcommand{\bbc}{\mathbb C}

\newcommand{\re}{\mathrm{Re}}

\begin{document}

\date{\today}


\thanks{\textbf{Acknowledgment.} The work of S.-Y. Ha was supported by National Research Foundation of Korea(NRF-2020R1A2C3A01003881).
The work of H. Park was supported by Basic Science Research Program through the National Research Foundation of Korea(NRF) funded by the Ministry of Education (2019R1I1A1A01059585) }

\begin{abstract}
We study emergent behaviors of the Lohe hermitian sphere(LHS) model with a time-delay for a homogeneous ensemble. The LHS model is a complex counterpart of the Lohe sphere(LS) aggregation model on the unit sphere in Euclidean space, and it describes the aggregation of particles on the unit hermitian sphere in $\bbc^d$ with $d \geq 2$, Recently it has been introduced by two authors of this work as a special case of the Lohe tensor model \cite{H-P1}. When the coupling gain pair satisfies a specific linear relation, namely the Stuart-Landau(SL) coupling gain pair, it can be embedded into the LS model on $\bbr^{2d}$. In this work, we show that if the coupling gain pair is close to the SL coupling pair case, the dynamics of the LHS model exhibits an emergent aggregate phenomenon via the interplay between time-delayed interactions and nonlinear coupling between states. For this, we present several frameworks for complete aggregation and practical aggregation in terms of initial data and system parameters using the Lyapunov functional approach. 
\end{abstract}

\maketitle \centerline{\date}

%\tableofcontents

\section{Introduction} \label{sec:1}
\setcounter{equation}{0}
Emergent dynamics of a many-body system is ubiquitous in classical and quantum systems, e.g., aggregation of bacteria \cite{T-B-L, T-B}, flocking of birds \cite{A-B-F}, schooling of fishes, synchronization of fireflies and neurons \cite{B-B, Pe, Wi1, Wi2} and hand clapping of people in a concert hall, etc. For surveys and books, we refer to \cite{A-B, A-B-F, B-H, D-B1, H-K-P-Z, P-R, St, VZ, Wi1}. In this paper, we continue  studies begun in \cite{B-H-P, H-P2} on the emergent dynamics of the LHS model. The LHS model corresponds to the complex counterpart of the Lohe sphere(LS) model which has been extensively studied in previous literature \cite{C-H3, J-C, Lo-2, M-T-G, T-M, Zhu}. The LHS model is the first-order aggregation model describing continuous-time dynamics of particle's position on the hermitian unit sphere $\mathbb{HS}^{d-1} := \Big \{ z = ([z]_1, \cdots, [z]_d)\in\bbc^d:~\|z \| := \sqrt{\sum_{\alpha=1}^{d} |[z]_\alpha|^2} = 1 \Big \}$ with $d \geq 1$. Here we denote the $\alpha$-th component of the complex vector $z\in\bbc^d$ as $[z]_\alpha$ which is consistent with earlier notation in \cite{H-P1}. As a warm up for our discussion, we briefly introduce the LHS model with time-delayed interactions. 

Let $z_j = ([z_j]_1, \cdots, [z_j]_d) \in \bbc^d$ be a position of the $j$-th Hermitian Lohe particle on the Hermitian unit sphere, and interaction weight between the $j$-th and $k$-th particle is denoted by the real value $a_{jk} \in \bbr$.  Then, the temporal dynamics of $z_j$ is governed by the Cauchy problem to the LHS model with a uniform time-delay $\tau > 0$:
\begin{equation} \label{A-1}
\begin{cases} 
\displaystyle \dot{z}_j =\Omega_j z_j + \displaystyle\frac{\kappa_0}{N}\sum_{k \neq j}a_{jk}\big(\langle z_j, z_j \rangle z_k^{\tau} - \langle z_k^{\tau}, z_j \rangle z_j) + \frac{\kappa_1}{N}\sum_{k \neq j}a_{jk}\big(\langle z_j, z_k^{\tau} \rangle - \langle z_k^{\tau} , z_j \rangle) z_j, \quad t > 0, \vspace{0.2cm} \\
\displaystyle z_j(t) = \varphi_j(t) \in \mathbb{HS}^{d-1}, \quad -\tau \leq t \leq 0, \quad j \in {\mathcal N} := \{1, \cdots, N \},
\end{cases}
\end{equation}
where $z_k^{\tau}(t) := z_k(t-\tau)$, $\varphi_j = \varphi_j(t)$ is a bounded continuous function of $t$, $\Omega_j$ is $d\times d$ skew-Hermitian matrix and $(a_{ik}) \in \bbr^{N \times N}$ is a symmetric matrix whose components are all positive. Before we continue further, we introduce
\[
\langle w, z \rangle := \sum_{\alpha = 1}^{d} [\bar{w}]_\alpha [z]_{\alpha}, \quad  \|z \| := \sqrt{ \langle z, z \rangle}, \quad \bar{w} = (\overline{[w]_1}, \cdots, \overline{[w]_d}).
\]
The global well-posedness of system \eqref{A-1} is guaranteed by the local well-posedness by the standard Cauchy-Lipschitz theory in \cite{Hale, Kuang} and a priori estimate in Lemma \ref{L2.1}. In the absence of time-delay with $\tau = 0$, emergent dynamics of the LHS model was investigated in \cite{B-H-P, H-P2} in which several sufficient frameworks were proposed for complete and practical aggregations.  In this paper, we are interested in the following simple question: \newline
\begin{quote}
``Under what conditions on system parameters $\kappa_0, \kappa_1, \tau$, network topology $(a_{ij})$ and initial data set $\{ \varphi_j \}$, can we verify the emergence of collective behaviors of the LHS with time-delay?"
\end{quote}

\vspace{0.2cm}

This question has been addressed in other low-dimensional aggregation models, to name a few, the Lohe sphere model \cite{C-H0, C-H2}, the Lohe matrix model \cite{HKKPS}. Throughout the paper, we set 
\[ Z := (z_1, \cdots, z_N), \quad {\mathcal D}(Z) := \max_{1 \leq i,j \leq N} \|z_i - z_j \|. \]
Next, we recall several concepts on the emergent dynamics in the following definition
\begin{definition} \label{D1.1}
Let $\{z_i \}$ be a global solution to \eqref{A-1}. 
\begin{enumerate}
\item 
Complete aggregation occurs asymptotically if the ensemble diameter ${\mathcal D}(Z)$ tends to zero asymptotically: 
\[ \lim_{t \to \infty} {\mathcal D}(Z(t)) = 0. \]
\item
Practical aggregation (with respect to time-delay) occurs asymptotically if the ensemble diameter ${\mathcal D}(Z)$ satisfies
\[ \lim_{\tau \to 0+} \limsup_{t \to \infty} {\mathcal D}(Z(t)) = 0.     \]
\end{enumerate}
\end{definition}
Then, it is easy to see that practical aggregation implies complete aggregation. In the absence of time-delay $\tau = 0$, emergent dynamics for system \eqref{A-1} has been extensively studied in \cite{H-P2} (see Section \ref{sec:2.3}). Thus, main point of this paper is to see the effect of time-delayed interactions in the emergent dynamics of \eqref{A-1}. For notational simplicity, we set 
\[
 \max_{i,j} := \max_{1 \leq i, j \leq N}, \quad  \min_{i,j} := \min_{1 \leq i, j \leq N},\quad \sum_{k \neq j}:=\sum_{\substack{k,j =1\\k\neq j}}^N. \]
The main results of this paper are threefold. First, we consider the following setting:
\[ a_{ik} \equiv 1, \quad \Omega_j = 0, \quad \forall~i, k \in {\mathcal N}. \] 
In this case, system \eqref{A-1} becomes 
\begin{align}\label{A-2}
\begin{cases}
\displaystyle \dot{z}_j=  \frac{\kappa_0}{N}\sum_{k \neq j}(\langle z_j, z_j \rangle z_k^{\tau} - \langle z_k^{\tau}, z_j \rangle z_j) + \frac{\kappa_1}{N}\sum_{k \neq j}(\langle z_j, z_k^{\tau} \rangle - \langle z_k^{\tau} , z_j \rangle) z_j, \\
\displaystyle z_j(t) = \varphi_j(t)\in \mathbb{HS}^{d-1}, \quad -\tau \leq t \leq 0.
\end{cases}
\end{align}
When the coupling gain pair $(\kappa_0, \kappa_1)$ is close to the SL coupling gain pair, i.e., 
\[   \tilde{\kappa} := \frac{\kappa_0}{2} + \kappa_1, \quad  |\tilde{\kappa}| \ll 1, \]
system \eqref{A-2} can be rewritten as follows (see Section \ref{sec:3}):
\begin{align}\label{A-3}
\begin{cases}
\displaystyle \dot{z}_j= \frac{\kappa_0}{N}\sum_{k \neq j} \Big(\langle z_j, z_j \rangle z_k^{\tau} - \mathrm{Re}(\langle z_k^{\tau}, z_j \rangle) z_j \Big) + \frac{\tilde{\kappa}}{N}\sum_{k \neq j} \Big(\langle z_j, z_k^{\tau} \rangle - \langle z_k^{\tau} , z_j \rangle \Big) z_j,  \vspace{0.2cm} \\
\displaystyle z_j(t) = \varphi_j(t) \in  \mathbb{HS}^{d-1}, \quad -\tau \leq t \leq 0.
\end{cases}
\end{align}
Our first set of results is concerned with the complete aggregation of \eqref{A-3} (see Theorem \ref{T3.1} and Theorem \ref{T3.2}). We assume that system parameters and initial data satisfy
\[ \kappa_0 > 0, \quad |\tilde \kappa| \ll \kappa_0, \quad \tau \ll 1, \quad  N \geq 3, \quad \sup_{-\tau \leq t \leq 0} {\mathcal D}(Z(t)) \ll 1. \]
For the complete aggregation, we introduce a Lyapunov functional:
\[
\mathcal{E}_{ij}(t) := \|z_i(t) - z_j(t)\|^2 + \gamma \int_{t-\tau}^{t}\|z_i(s) - z_j(s)\|^2 ds.
\]
Then, it is easy to see that it satisfies the standard energy estimate (see Section \ref{sec:3.2.2}):
\[  \mathcal{E}_{ij}(t)  + \beta\int_0^t\|z_i(s)-z_j(s)\|^2ds \leq \mathcal{E}_{ij}(0), \quad \forall~t > 0.   \]
By Barbalat's lemma \cite{Ba}, the above estimate leads to complete aggregation (see Theorem \ref{T3.2}): 
\[  \lim_{t \to \infty} \|z_i(t)-z_j(t)\| = 0. \]
Now, our second set of result is concerned with the practical aggregation with respect to time-delay (Theorem \ref{T4.1}). We assume that  system parameter and initial data satisfy
\[ 2|\kappa_1|<\kappa_0 \quad  \max_{i,j}  \Big( 1-  \langle z^0_i, z^0_j \rangle \Big) <1-\frac{2|\kappa_1|}{\kappa_0}. \]
Then, a practical aggregation emerges:
\[ \lim_{\tau\searrow0}\limsup_{t\to\infty} \max_{i,j} \Big(1-  \langle z_i(t), z_j(t) \rangle \Big) = 0. \]
Our final set of result is concerned with the practical aggregation with respect to both time-delay and non-identical free flow matrix $\Omega_j$ (Theorem \ref{T4.2}). For system \eqref{A-1}, we assume that system parameter satisfies
\begin{equation} \label{New-ini}
\max_{i, j}|1-\langle z_i^0, z_j^0\rangle|<1-\frac{2\sum_{k=1}^N|a_{ik}-a_{jk}|}{\sum_{k=1}^N(a_{ik}+a_{jk})}.
\end{equation}
Then, a practical aggregation emerges:
\[
\lim_{\kappa_0\to\infty}\lim_{\tau\searrow0}\limsup_{t\to\infty}\max_{i,j} \Big(1-  \langle z_i(t), z_j(t) \rangle \Big) = 0.
\]
Note that although we imposed the initial condition on $\varphi_j(t)$ for a time-strip $-\tau\leq t\leq 0$, we require that  the initial condition depends on the initial data at $t = 0$ for practical aggregation estimate in Theorem \ref{T4.2}. 

\vspace{0.5cm}

The rest of paper is organized as follows. In Section \ref{sec:2}, we present conservation laws for the LHS model with time-delay, its reduction to other aggregation models, and review previous results on the emergent dynamics for the LHS model without time-delay and LS model with a time-delay. In Section \ref{sec:3}, we provide a sufficient framework for the complete aggregation when the coupling gain pair is close to that of SL coupling gain pair. In Section \ref{sec:4}, we provide a sufficient framework leading to the practical aggregation under a general setting. Finally, Section \ref{sec:5} is devoted to a brief summary of main results and some open problems which will be left for a future work. 
 
 \section{Preliminaries} \label{sec:2}
\setcounter{equation}{0}
In this section, we discuss two conservation laws of the LHS model with time-delay and its reduction to other aggregation models, and review previous results on the emergent dynamics for the LHS model.\newline

\subsection{Conservation laws} \label{sec:2.1}
In this subsection, we study conservation laws associated with \eqref{A-1}. 
\begin{lemma}\label{L2.1}
\emph{(Conservation of modulus)}
Let $\{z_j\}$ be a global solution to \eqref{A-1}. Then, the modulus of $z_j(t)$ satisfies
\[ \|z_j(t)\|=1, \quad t \geq 0, \quad j \in {\mathcal N}. \]
i.e., the hermitian Lohe sphere $\mathbb{HS}^{d-1}$ is positively invariant set for \eqref{A-1}.
\end{lemma}
\begin{proof} We use $(a_{jk}) \in \bbr^{N \times N}$ and sesquilinearity of the inner product to find
\begin{align}
\begin{aligned} \label{B-2}
&\frac{d}{dt}\|z_j\|^2 =\langle \dot{z}_j, z_j\rangle+\langle z_j, \dot{z}_j\rangle\\
& \hspace{0.2cm} = \left\langle \Omega_j z_j + \displaystyle\frac{\kappa_0}{N}\sum_{k \neq j} a_{jk} \big(\langle z_j, z_j \rangle z_k^{\tau} - \langle z_k^{\tau}, z_j \rangle z_j) + \frac{\kappa_1}{N}\sum_{k \neq j} a_{jk} \big(\langle z_j, z_k^{\tau} \rangle - \langle z_k^{\tau} , z_j \rangle) z_j,~z_j \right\rangle \\
&\hspace{0.4cm}+ \left\langle z_j, ~\Omega_j z_j + \displaystyle\frac{\kappa_0}{N}\sum_{k \neq j} a_{jk} \big(\langle z_j, z_j \rangle z_k^{\tau} - \langle z_k^{\tau}, z_j \rangle z_j) + \frac{\kappa_1}{N}\sum_{k \neq j} a_{jk} \big(\langle z_j, z_k^{\tau} \rangle - \langle z_k^{\tau} , z_j \rangle) z_j \right\rangle \\
\end{aligned}
\end{align}
\begin{align*}
\begin{aligned}
&\hspace{0.2cm} = \langle \Omega_j z_j , z_j \rangle + \langle z_j, \Omega_j z_j \rangle \\
&\hspace{0.4cm}+ \displaystyle\frac{\kappa_0}{N}\sum_{k \neq j}  a_{jk}  \big(\langle z_j, z_j \rangle \langle z_k^{\tau}, z_j \rangle - \overline{\langle z_k^{\tau}, z_j \rangle} \langle z_j, z_j \rangle) + \displaystyle\frac{\kappa_0}{N}\sum_{k \neq j}  a_{jk}  \big(\langle z_j, z_j \rangle \langle z_j, z_k^{\tau} \rangle - \langle z_k^{\tau}, z_j \rangle \langle z_j, z_j \rangle)  \\
&\hspace{0.4cm}+ \displaystyle\frac{\kappa_1}{N}\sum_{k \neq j}  a_{jk}  \big(\overline{\langle z_j, z_k^{\tau} \rangle} \langle z_j, z_j \rangle - \overline{\langle z_k^{\tau}, z_j \rangle} \langle z_j, z_j \rangle) + \displaystyle\frac{\kappa_1}{N}\sum_{k \neq j}  a_{jk} \big(\langle z_j, z_k^{\tau} \rangle \langle z_j, z_j \rangle - \langle z_k^{\tau}, z_j \rangle \langle z_j, z_j \rangle) \\ 
&\hspace{0.2cm} =: \sum_{i=1}^{6} {\mathcal I}_{1i}.   
\end{aligned}
\end{align*}

\vspace{0.5cm}

Below, we estimate the terms ${\mathcal I}_{1i}$ with $1\leq i\leq 6$ one by one. \newline

\noindent $\bullet$~Case A (Estimates on ${\mathcal I}_{11} + {\mathcal I}_{12}$): Note that 
\[
\langle \Omega_j z_j , z_j \rangle = \langle z_j, \Omega^\dagger_j z_j  \rangle
= \langle z_j, -\Omega_j z_j \rangle
= - \langle z_j, \Omega_j z_j \rangle
= - \overline{\langle \Omega_j z_j, z_j \rangle}.
\] Thus $\langle \Omega z_j, z_j \rangle$ is purely imaginary. This yields 
\[
{\mathcal I}_{11} + {\mathcal I}_{12} = \langle \Omega_j z_j , z_j \rangle + \langle z_j, \Omega_j z_j \rangle = - \overline{\langle \Omega_j z_j, z_j \rangle} + \overline{\langle \Omega_j z_j, z_j \rangle} = 0.
\]

\vspace{0.5cm}

\noindent $\bullet$~Case B (Estimates on ${\mathcal I}_{13} + {\mathcal I}_{14}$):~We use $ \langle z_j, z_k^{\tau} \rangle = \overline{\langle z_k^{\tau}, z_j \rangle}$ to see that 
\[  {\mathcal I}_{13} +  {\mathcal I}_{14}  = 0.   \]

\vspace{0.5cm}

\noindent $\bullet$~Case C (Estimates on ${\mathcal I}_{15} + {\mathcal I}_{16}$):~Similar to Case B, one has 
\[  {\mathcal I}_{15} +  {\mathcal I}_{16}  = 0.   \]
Finally we combine all the estimates in Cases A, B, and C to obtain
\[
\frac{d}{dt}\|z_j(t)\|^2=0, \quad \forall~t > 0,\quad j  \in {\mathcal N}.
\]
This yields
\[
\|z_j(t)\|=\|z_j(0)\|=\|\varphi_j(0)\|=1.
\]
\end{proof}
\begin{lemma}\label{L2.2}
\emph{(Propagation of real-valuednesss)}
Suppose that $\{\Omega_j\}$ and initial data set $\{ \varphi_j \}$ satisfy the relations:
\[
\Omega_j \in\bbr^{d\times d},\quad \Omega_j^T=-\Omega_j,\quad \varphi_j(t)\in\bbr^d,\quad \|\varphi_j(t)\|=1
\]
for all $j\in\mathcal{N}$ and $-\tau\leq t\leq 0$, and let $\{z_j\}$ be a solution to \eqref{A-1}. Then $z_j$ is a real-valued state, i.e.,
\[
\mathrm{Im}([z_j(t)]_\alpha)=0, \quad \forall t \geq 0,~~\alpha \in \{1, \cdots, d\}, \quad j \in {\mathcal N}.
\]
\end{lemma}
\begin{proof} This follows from the standard uniqueness theory of time-delayed ordinary differential equations \cite{Hale, Kuang}
\end{proof}
\subsection{Reduction to aggregation models} \label{sec:2.2}
In this subsection, we discuss the reductions of \eqref{A-1} to the Lohe sphere model and the Kuramoto model. Suppose that initial data set $\{\varphi_j \}$ satisfy 
\[  \varphi_j(t)\in\bbr^d,\quad \|\varphi_j(t)\|=1,\]
for all $j\in\mathcal{N}$ and $-\tau\leq t\leq 0$. Then, it follows from Lemma \ref{L2.1} and Lemma \ref{L2.2} that 
\[ z_j(t)\in\bbs^{d-1}\subset\bbr^d. \]
In this case, the coupling terms in the R.H.S. of \eqref{A-1} become
\[ \langle z_j, z_j\rangle z_k^\tau-\langle z_k^\tau, z_j\rangle z_j=\|z_j\|^2 z_k^\tau-\langle z_k^\tau, z_j\rangle z_j, \quad 
(\langle z_j, z_k^\tau\rangle-\langle z_k^\tau, z_j\rangle)z_j=0.
\]
We set 
\[  x_j(t) := z_j(t), \quad j \in {\mathcal N},~t \geq 0. \]
Then the real-valued state $x_j \in \bbr^d$ satisfies the LS model with time-delay \cite{C-H2}:
\begin{align}\label{B-2-1}
\begin{cases}
\displaystyle \dot{x}_j=\Omega_j x_j+\displaystyle\frac{\kappa_0}{N}\sum_{k\neq j}(\|x_j\|^2 x_k^\tau-\langle x_k^\tau, x_j\rangle x_j),\\
\displaystyle x_j(t)=\varphi_j(t)\in\bbs^{d-1}\subset\bbr^d,\quad -\tau\leq t\leq 0,
\end{cases}
\end{align}
where $\Omega_j$ is a $d\times d$ skew-symmetric matrix for all $j$. Emergent dynamics of \eqref{B-2-1} has been studied in \cite{C-H0}. To see the reduction to the Kuramoto model, we also set 
\begin{equation} \label{B-3}
d=2,\quad x_j:=\begin{bmatrix}
\cos\theta_j\\
\sin\theta_j
\end{bmatrix},\quad
\varphi_j:=\begin{bmatrix}
\cos\alpha_j\\
\sin\alpha_j
\end{bmatrix},\quad
\Omega_j:=\begin{bmatrix}
0&-\nu_j\\
\nu_j&0
\end{bmatrix}, \quad \kappa_0 = \kappa.
\end{equation}
Again, we substitute the ansatz \eqref{B-3} into \eqref{B-2} to derive the Kuramoto model with time-delay \cite{HLLP, HNP}:
\begin{equation}\label{B-4}
\begin{cases}
\displaystyle \dot{\theta}_j=\nu_j +\displaystyle\frac{\kappa}{N}\sum_{k\neq j}\sin(\theta_k^\tau-\theta_j), \quad t > 0,\\
\displaystyle \theta_j(t)=\alpha_j(t),\quad -\tau\leq t\leq 0,\quad j\in {\mathcal N}.
\end{cases}
\end{equation}
For the emergent dynamics of \eqref{B-4}, we refer to \cite{B-C-M,  C-S, D-X, D-B, Ku2}. In summary, one has  the following diagram:
\[ \mbox{LHS model} \quad \Longrightarrow \quad \mbox{Lohe sphere model}  \quad \Longrightarrow \quad \mbox{Kuramoto model}. \]
We also refer to \cite{B-C-S, D-F-M-T, De, H-Kim-R, HKKPS, Kim,Lo-0, Lo-1, Lo-2} for other Lohe type matrix models.

%\begin{center}
%\begin{tabular}{c c c c c}
%&&Subsystem A type \eqref{C-4} &$\rightarrow$&Lohe Sphere model type \eqref{C-2}\\
%&$\nearrow$&&$\searrow$&$\downarrow$\\
%LHS model type \eqref{C-1}&&&&Kuramoto type \eqref{C-3}$$\\
%&$\searrow$&&$\nearrow$&\\
%&&Subsystem B type \eqref{C-5}
%\end{tabular}
%\end{center}
%\vspace{0.5cm}

\subsection{Previous results} \label{sec:2.3} In this subsection, we present two results on the emergent dynamics of the  LHS model without a time-delay and the Lohe sphere model with time-delay which correspond to the special cases for system \eqref{A-1}. \newline

First, we consider the LHS model with zero time-delay case with $\tau = 0$ over the complete network with $a_{ik} = 1$. Under these setting, system \eqref{A-1} becomes 
\begin{equation}\label{B-4}
\begin{cases}
\dot{z}_j=\Omega_j z_j+\displaystyle\frac{\kappa_0}{N}\sum_{k=1}^N\big(\langle z_j, z_j\rangle z_k-\langle z_k, z_j\rangle z_j\big)+\frac{\kappa_1}{N}\sum_{k=1}^N\big(\langle z_j, z_k\rangle-\langle z_k, z_j\rangle\big)z_j,~~t > 0, \\
z_j(0)=z_j^{in}\in \mathbb{HS}^{d-1},\quad  j \in \mathcal{N}.
\end{cases}
\end{equation}

For emergent dynamics of \eqref{B-4}, we introduce an order parameter as a modulus of $z_c$ and state diameter:
\begin{equation} \label{B-4-1}
\rho := \|z_c\|, \quad \mathcal{D}(Z) := \max_{i,j}\|z_i - z_j\|.
\end{equation}
On the other hand, we consider system \eqref{B-4-1} with a zero free flow:
\begin{equation} \label{B-5}
\begin{cases}
\displaystyle \dot{w}_j=  \kappa_{0} \Big(w_c\langle w_j, w_j \rangle-w_j \langle w_c. w_j \rangle \Big)+\kappa_1 \Big(\langle{w_j, w_c}\rangle- \langle w_c, w_j\rangle \Big) w_j, \quad t > 0, \vspace{0.2cm} \\
\displaystyle w_j(0)  = z_j^{in}, \quad j \in \mathcal{N},
\end{cases}
\end{equation}
where $w_c:=\displaystyle\frac{1}{N}\sum_{k=1}^Nw_k$. \newline

Then the emergence of complete aggregation and solution splitting property of \eqref{B-4} can be summarized in the following proposition.  
\begin{proposition} 
\cite{H-P2} Suppose that coupling gains, free flows and initial data satisfy
\[
N \geq 3, \quad 0<\kappa_1<\frac{1}{4}\kappa_0,  \quad \rho^{in} > \frac{N-2}{N}, 
\]
and let $\{z_j \}$ be a global solution to \eqref{B-4}. Then, the following assertions hold.
\begin{enumerate}
\item
Complete aggregation emerges asymptotically:
\[ \lim_{t \to \infty} {\mathcal D}(Z(t)) = 0. \]
\item
Solution splitting property holds:
\[ z_j=e^{\Omega t} w_j, \quad j \in\mathcal{N},   \]
where $w_j$ is a solution to \eqref{B-5}.
\end{enumerate}
\end{proposition}
\begin{proof} For a proof, we refer to \cite{H-P2}.
\end{proof} 
Second, we consider the Lohe sphere model on the unit sphere in $\bbr^{d}$ under the influence of time-delay:
\begin{equation}\label{B-6}
\begin{cases}
\displaystyle \dot{x}_j=\Omega x_j+\displaystyle\frac{\kappa}{N}\sum_{k\neq j}\left(\|x_j\|^2 x_k^\tau-\langle x_k^\tau, x_j\rangle x_j\right), \quad t > 0,~~j \in {\mathcal N}, \vspace{0.2cm} \\
\displaystyle x_i(t)=\varphi_i(t)\in\bbs^{d-1},\quad -\tau\leq t\leq0,
\end{cases}
\end{equation}
For an emergent dynamics, we introduce a modified ensemble diameter as follows:
\begin{equation} \label{B-7}
\mathcal{D}^{0,\tau}(t) : = \displaystyle\max_{i,j}\|z_i(t) - z_j^\tau(t)\|.
\end{equation}
\begin{proposition}
\emph{\cite{C-H2}}
Suppose that the system parameters and initial data satisfy 
\begin{align*}
\begin{aligned}
& N \geq 3, \quad  \kappa >0, \quad \tau < \displaystyle\frac{1}{8(d\|\Omega\|_{\infty} +2\kappa)}, \\
& \|\varphi(t)\|=1, \quad t \in [-\tau, 0], \quad \sup_{-\tau \leq t \leq 0} \mathcal{D}(\varphi(t))< \frac{1}{8}, 
\end{aligned}
\end{align*}
and let $\{x_j \}$ be a global solution to $\eqref{B-6}$. Then, we have
\[
\lim_{t \to \infty} {\mathcal D}(X(t)) = 0.
\]
\end{proposition} 
\begin{proof}
For a proof, we refer to Theorem 3.1 of \cite{C-H2}.
\end{proof}

\section{Emergence of complete aggregation} \label{sec:3}
\setcounter{equation}{0}
In this section, we provide an emergent dynamics of \eqref{A-1} under the following setting:
\[ a_{ik} \equiv 1, \quad i, k \in {\mathcal N} \quad \mbox{and} \quad \Omega = 0.\] 
Note that this case corresponds to the same free flow and complete network topology. Then system \eqref{A-1} becomes 
\begin{align}\label{C-1}
\begin{cases}
\dot{z}_j= \displaystyle  \frac{\kappa_0}{N}\sum_{k \neq j}(\langle z_j, z_j \rangle z_k^{\tau} - \langle z_k^{\tau}, z_j \rangle z_j) + \frac{\kappa_1}{N}\sum_{k \neq j}(\langle z_j, z_k^{\tau} \rangle - \langle z_k^{\tau} , z_j \rangle) z_j, \\
z_j(t) = \varphi_j(t)\in \mathbb{HS}^{d-1}, \quad -\tau \leq t \leq 0.
\end{cases}
\end{align}
In the following two subsections, we study complete aggregation in which coupling gains satisfy the following relations:
\begin{align*}
\begin{aligned}
& \kappa_1 + \frac{\kappa_0}{2} = 0 \quad  \mbox{(Stuart-Landau(SL) coupling gain pair)}, \\
& 0 < \Big| \kappa_1 + \frac{\kappa_0}{2} \Big| \ll 1 \quad  \mbox{(Close-to-SL coupling gain pair)}.
\end{aligned}
\end{align*}
In Section 2.3 of \cite{B-H-P}, authors reduced the vector version of Stuart-Landau model to the LHS model with the special pair of coupling gains.  From this process, Stuart-Landau(SL) coupling gain pair and close-to SL coupling gain pair were naturally obtained.
\subsection{SL coupling gain pair} \label{sec:3.1}
In this subsection, we consider the emergent behavior of \eqref{C-1} for the Stuart-Landau gain pair. In this case, the coupling term can be simplified as follows: on $\mathbb{HS}^{d-1}$,
\begin{align}
\begin{aligned} \label{C-1-1}
& \kappa_0 (\langle z_j, z_j \rangle z_k^{\tau} - \langle z_k^{\tau}, z_j \rangle z_j) + \kappa_1 (\langle z_j, z_k^{\tau} \rangle - \langle z_k^{\tau} , z_j \rangle) z_j \\
& \hspace{0.5cm} = \kappa_0 \Big[  z_k^{\tau} - \langle z_k^{\tau}, z_j \rangle z_j - \frac{1}{2}  (\langle z_j, z_k^{\tau} \rangle - \langle z_k^{\tau} , z_j \rangle) z_j                       \Big] \\
& \hspace{0.5cm} = \kappa_0 \left[ z^{\tau}_k-\frac{1}{2}\Big(  \langle z^{\tau}_k, z_j \rangle + \langle z_j, z^{\tau}_k \rangle \Big)z_j\right ] \\
& \hspace{0.5cm} =  \kappa_0 \big( z_k^{\tau} - \mathrm{Re}( \langle z_k^{\tau}, z_j \rangle) z_j\big).
\end{aligned}
\end{align}               
Finally, we combine \eqref{C-1} and \eqref{C-1-1} to get
\begin{align}\label{C-2}
\begin{cases}
\dot{z}_j= \displaystyle\frac{\kappa_0}{N}\sum_{k \neq j}\big( z_k^{\tau} - \mathrm{Re}( \langle z_k^{\tau}, z_j \rangle) z_j\big), \quad t > 0 ,\\
z_j(t) = \varphi_j(t)\in\bbc^d, \quad -\tau \leq t \leq 0.
\end{cases}
\end{align}

\begin{theorem} \label{T3.1}
Suppose system parameters and initial data set $\varphi_j$ satisfy
\[ \kappa_0 > 0, \quad N \geq 3,\quad i\in \mathcal{N}, \quad\tau < \displaystyle\frac{1}{16\kappa_0}, \quad \|\varphi_j\| =1, \quad D(\varphi(t)) < \displaystyle\frac{1}{8}, \quad t \in [-\tau , 0],
\]
and let $\{ z_j \}$ be a global solution to \eqref{C-2}. Then, the complete aggregation emerges asymptotically:
\[ \lim_{t \to \infty} {\mathcal D}(Z(t)) = 0. \]
\end{theorem}
\begin{proof}
We leave its proof in Section \ref{sec:3.1.2}.
\end{proof}

\subsubsection{Basic a priori estimates} \label{sec:3.1.1} In the part, we provide four lemmas for the emergent dynamics of \eqref{C-2} following the strategy in \cite{C-H2}. 
\begin{lemma} \label{L3.1}
Let $\{z_j\}$ be a global solution to \eqref{C-2}. Then we have
\begin{align*}
\frac{d}{dt}\|z_i-z_j^s\|^2&\leq2\kappa_0\re\langle z_c^\tau-z_c^{\tau+s}, z_i-z_j^s\rangle-\kappa_0\|z_i-z_j^s\|^2(\re\langle z_c^\tau,z_i\rangle+\re\langle z_c^{\tau+s}, z_j^s\rangle)\\
&\hspace{0.5cm} -\frac{2\kappa_0}{N}\left(\re\langle z_i-z_j^s,z_i^\tau-z_j^{\tau+s}\rangle-\|z_i-z_j^s\|^2\right),
\end{align*}
for all $j\in\mathcal{N}$ and $t\geq s+\tau$.
\end{lemma}
\begin{proof} We set 
\[  z_j^{s}(t) = z_j(t-s), \quad j \in {\mathcal N}. \]
Then, it satisfies 
\begin{equation} \label{C-3}
\dot{z}_j^s= \displaystyle\frac{\kappa_0}{N}\sum_{k \neq j}\big( z_k^{\tau+s} - \mathrm{Re}( \langle z_k^{\tau+s}, z_j^s \rangle) z_j^s\big).
\end{equation}
It follows from $\eqref{C-2}_1$ and \eqref{C-3} that 
\begin{align}
\begin{aligned} \label{C-3-1}
\frac{d}{dt}(z_i-z_j^s)
&=\frac{\kappa_0}{N}\left(\sum_{k\neq i}(z_k^\tau-\mathrm{Re}(\langle z_k^\tau, z_i\rangle)z_i)-\sum_{k\neq j}(z_k^{\tau+s}-\mathrm{Re}(\langle z_k^{\tau+s},z_j^s\rangle)z_j^s)\right)\\
&=\kappa_0\left((z_c^\tau-\mathrm{Re}(\langle z_c^\tau, z_i\rangle)z_i)-(z_c^{\tau+s}-\mathrm{Re}(\langle z_c^{\tau+s}, z_j^s\rangle)z_j^s)\right)\\
&\hspace{0.2cm}-\frac{\kappa_0}{N}\left((z_i^\tau-\mathrm{Re}(\langle z_i^\tau, z_i\rangle)z_i)-(z_j^{\tau+s}-\mathrm{Re}(\langle z_j^{\tau+s},z_j^s\rangle)z_j^s)\right).
\end{aligned}
\end{align}
This yields
\begin{align}
\begin{aligned} \label{C-4}
&\frac{d}{dt}\|z_i-z_j^s\|^2 =2\mathrm{Re}\left\langle z_i-z_j^s, \frac{d}{dt}(z_i-z_j^s)\right\rangle\\
& \hspace{0.7cm} =2\kappa_0 \mathrm{Re}\left\langle z_i-z_j^s,z_c^\tau-\mathrm{Re}(\langle z_c^\tau, z_i\rangle)z_i \right\rangle
-2\kappa_0 \mathrm{Re}\left\langle z_i-z_j^s,z_c^{\tau+s}-\mathrm{Re}(\langle z_c^{\tau+s}, z_j^s\rangle)z_j^s \right\rangle\\
&\hspace{0.7cm} -\frac{2\kappa_0}{N}\mathrm{Re}\left\langle z_i-z_j^s,z_i^\tau-\mathrm{Re}(\langle z_i^\tau, z_i\rangle)z_i \right\rangle
+\frac{2\kappa_0}{N}\mathrm{Re}\left\langle z_i-z_j^s,z_j^{\tau+s}-\mathrm{Re}(\langle z_j^{\tau+s},z_j^s\rangle)z_j^s \right\rangle\\
&\hspace{0.7cm} =2\kappa_0 \mathrm{Re}\left\langle -z_j^s,z_c^\tau-\mathrm{Re}(\langle z_c^\tau, z_i\rangle)z_i \right\rangle
-2\kappa_0 \mathrm{Re}\left\langle z_i,z_c^{\tau+s}-\mathrm{Re}(\langle z_c^{\tau+s}, z_j^s\rangle)z_j^s \right\rangle\\
& \hspace{0.7cm} -\frac{2\kappa_0}{N}\mathrm{Re}\left\langle -z_j^s,z_i^\tau-\mathrm{Re}(\langle z_i^\tau, z_i\rangle)z_i \right\rangle
+\frac{2\kappa_0}{N}\mathrm{Re}\left\langle z_i,z_j^{\tau+s}-\mathrm{Re}(\langle z_j^{\tau+s},z_j^s\rangle)z_j^s \right\rangle\\
& \hspace{0.7cm} =2\kappa_0(\re\langle z_c^\tau,z_i\rangle\re\langle z_j^s, z_i\rangle+\re\langle z_c^{\tau+s}, z_j^s\rangle\re\langle z_i,z_j^s\rangle-\re\langle z_j^s, z_c^\tau\rangle-\re\langle z_i, z_c^{\tau+s}\rangle)\\
& \hspace{0.7cm}-\frac{2\kappa_0}{N}\big(\re\langle z_i^\tau,z_i\rangle\re\langle z_j^s,z_i\rangle+\re\langle z_i,z_j^s\rangle\re\langle z_j^{\tau+s},z_j^s\rangle-\re\langle z_j^s,z_i^\tau\rangle-\re\langle z_i, z_j^{\tau+s}\rangle\big).
\end{aligned}
\end{align}
On the other hand, we have
\begin{equation} \label{C-5}
\|z_i-z_j^s\|^2=2(1-\re\langle z_i, z_j^s\rangle), \quad \mbox{i.e.,} \quad \re\langle z_i, z_j^s\rangle=1-\frac{1}{2}\|z_i-z_j^s\|^2.
\end{equation}
We combine \eqref{C-4} and \eqref{C-5} to obtain
\begin{align}
\begin{aligned} \label{C-6}
&\frac{d}{dt}\|z_i-z_j^s\|^2v=2\kappa_0(\re\langle z_c^\tau,z_i\rangle+\re\langle z_j^{\tau+s}, z_j^s\rangle-\re\langle z_j^s, z_c^\tau\rangle-\re\langle z_i, z_c^{\tau+s}\rangle)\\
&\hspace{0.5cm}-\kappa_0\|z_i-z_j^s\|^2(\re\langle z_c^\tau,z_i\rangle+\re\langle z_c^{\tau+s}, z_j^s\rangle)\\
&\hspace{0.5cm}-\frac{2\kappa_0}{N}(\re\langle z_i^\tau,z_i\rangle+\re\langle z_j^{\tau+s},z_j^s\rangle-\re\langle z_j^s,z_i^\tau\rangle-\re\langle z_i, z_j^{\tau+s}\rangle)\\
&\hspace{0.5cm}+\frac{\kappa_0}{N}\|z_i-z_j^s\|^2(\re\langle z_i^\tau,z_i\rangle+\re\langle z_j^{\tau+s},z_j^s\rangle)\\
&=2\kappa_0\re\langle z_c^\tau-z_c^{\tau+s}, z_i-z_j^s\rangle-\kappa_0\|z_i-z_j^s\|^2(\re\langle z_c^\tau,z_i\rangle+\re\langle z_c^{\tau+s}, z_j^s\rangle)\\
&\hspace{0.5cm}-\frac{2\kappa_0}{N}\re\langle z_i-z_j^s,z_i^\tau-z_j^{\tau+s}\rangle+\frac{\kappa_0}{N}\|z_i-z_j^s\|^2(\re\langle z_i^\tau,z_i\rangle+\re\langle z_j^{\tau+s},z_j^s\rangle).
\end{aligned}
\end{align}
Finally, \eqref{C-6} and $|\langle z, w\rangle|\leq \|z\|\cdot\|w\|$ yield desired estimate.
\end{proof}

\begin{lemma}\label{L3.2}
Let $\{z_j\}$ be a global solution to \eqref{C-2}. Then we have following relation for suitable positive numbers $s,u,t$:
\begin{align*}
\begin{aligned}
&\Big|\|z_i(t) - z_j^s(t)\|^2 - \mathrm{Re}\langle z_i^{u}(t) - z_j^{u +s}(t), z_i(t) - z_j^s(t) \rangle\Big| \\
& \hspace{1cm} \leq 2u\kappa_0  \sup_{t-u < v < t}\Big(\|z_i(v) - z_j^s(v)\| + \|z_c^{\tau}(v) - z_c^{\tau +s}(v)\|\Big) \|z_i(t) -z_j^s(t)\| \\
&\hspace{1.2cm} +\frac{2u\kappa_0}{N}  \sup_{t-u < v < t} \Big(\|z_i(v) - z_j^s(v)\| + \|z_i^{\tau}(v) - z_j^{\tau +s}(v)\|\Big) \|z_i(t) -z_j^s(t)\|.
\end{aligned}
\end{align*}
\end{lemma}
\begin{proof} Note that 
\begin{align*}
&\left|\|z_i(t) - z_j^s(t)\|^2 - \mathrm{Re}\langle z_i^{u}(t) - z_j^{u +s}(t), z_i(t) - z_j^s(t) \rangle\right|\\
&\hspace{3cm}= \left|\mathrm{Re}\left(\|z_i(t) - z_j^s(t)\|^2 - \langle z_i^{u}(t) - z_j^{u +s}(t), z_i(t) - z_j^s(t) \rangle \right)\right|.
\end{align*}
We integrate \eqref{C-3-1} on the interval $[t-u, t]$ and take the inner product of the resulting relation and $ z_i(t) - z_j^s(t)$ as in \cite{C-H2} to find
\begin{align*}
\begin{aligned}
&\hspace{-1cm}\left|\mathrm{Re}\left(\|z_i(t) - z_j^s(t)\|^2 - \langle z_i^{u}(t) - z_j^{u +s}(t), z_i(t) - z_j^s(t) \rangle \right)\right|\\
&  \leq \left|\|z_i(t) - z_j^s(t)\|^2 - \langle z_i^{u}(t) - z_j^{u +s}(t), z_i(t) - z_j^s(t) \rangle\right| \\
&  \leq  2u\kappa_0  \sup_{t-u < v < t}\Big(\|z_i(v) - z_j^s(v)\| + \|z_c^{\tau}(v) - z_c^{\tau +s}(v)\|\Big) \|z_i(t) -z_j^s(t)\| \\
&\hspace{0.2cm} +\frac{2u\kappa_0}{N}  \sup_{t-u < v < t} \Big(\|z_i(v) - z_j^s(v)\| + \|z_i^{\tau}(v) - z_j^{\tau +s}(v)\|\Big) \|z_i(t) -z_j^s(t)\|.
\end{aligned}
\end{align*}
\end{proof}

\begin{lemma} \label{L3.3}
Let $\{z_j\}$ be a global solution to \eqref{C-2}. Then, the functional \eqref{B-7} satisfies 
\begin{align*}
\begin{aligned}
\frac{d}{dt} D^{0,\tau}(t) &\leq \kappa_0 \| z_c^{\tau} - z_c^{2\tau}\| - \frac{\kappa_0D^{0, \tau}(t)}{2}
\left(2 - \frac{D^{0, \tau}(t)^2}{2} -\frac{D^{0, \tau}(t-\tau)^2}{2}\right) \\
&\hspace{0.5cm} +\frac{4\kappa_0^2\tau(N+1)}{N^2}\left(\sup_{t-2\tau< v < t}D^{0, \tau}(v)\right).
\end{aligned}
\end{align*}
\end{lemma}
\begin{proof} In Lemma \ref{L3.2}, we set $s =\tau$ and take $s, u = \tau$. Since inequalities in Lemma \ref{L3.2} and Lemma \ref{L3.3} are similar to estimates in Lemma 4.1 and 4.2 in \cite{C-H2}, we can derive the same result. The only difference is that we have terms involving real parts, but it can be estimated in the same way as \cite{C-H2} since 
\begin{align*}
1 - \mathrm{Re}(\langle z_i , z_c \rangle) &= \mathrm{Re}(1 - \langle z_i , z_c \rangle) = \frac{1}{N}\sum_{k=1}^{N}\mathrm{Re}(1 - \langle z_i, z_k \rangle)
= \frac{1}{N}\sum_{k=1}^{N}\frac{\|z_i - z_k\|^2}{2} \leq \frac{D^{0,\tau}(t)^2}{2}.
\end{align*}
\end{proof}
We set 
\begin{equation} \label{C-7}
\Delta_{z_j}^\tau(t) = \|z_j(t) - z_j^{\tau}(t)\|.
\end{equation}
Then we have an estimate for $\Delta_{z_j}^\tau$ as following lemma.
\begin{lemma} \label{L3.4}
Let $\{z_j\}$ be a global solution to \eqref{C-2}. Then, the functional $\Delta_{z_j}^\tau$ satisfies 
\[
\Delta_{z_j}^\tau(t) \leq 2\kappa_0\tau \left( \frac{N-1}{N}\right)
\]
\end{lemma}
\begin{proof} Note that 
\[
z_j(t) - z_j^{\tau}(t) = z_j(t) - z_j(t-\tau) = \int_{t-\tau}^{t} \dot{z}_j (s)ds,
\]
This yields
\begin{align*}
\left\|\int_{t-\tau}^{t} \dot{z}_j (s)ds\right\| &= \left\|\int_{t-\tau}^{t} \left(\displaystyle\frac{\kappa_0}{N}\sum_{k \neq j}\big( z_k^{\tau} - \mathrm{Re}( \langle z_k^{\tau}, z_j \rangle) z_j\big)\right)ds\right\|\\
&\leq \int_{t-\tau}^{t} \displaystyle\frac{\kappa_0}{N}\sum_{k \neq j}\big( z_k^{\tau} - \mathrm{Re}( \langle z_k^{\tau}, z_j \rangle) z_j\big)\|ds \\
& \leq \int_{t-\tau}^{t} \displaystyle\frac{\kappa_0}{N}\sum_{k \neq j}\big( \|z_k^{\tau}\| + |\mathrm{Re}( \langle z_k^{\tau}, z_j \rangle)|\cdot \|z_j\|\big)ds \\
& \leq \int_{t-\tau}^{t}\displaystyle\frac{\kappa_0}{N}\sum_{k \neq j}2ds=2\kappa_0\tau \left( \frac{N-1}{N}\right).
\end{align*}
\end{proof}
Now we are ready to provide a proof of our first main result. 

{\subsubsection{Proof of Theorem \ref{T3.1}} \label{sec:3.1.2} In this part, we present our first result on the complete aggregation by combining all the estimates in Lemma \ref{L3.1} - Lemma \ref{L3.4} in two steps. We will briefly sketch the proof, since in the next section, we will provide more general statement ans its proof.  \vspace{0.2cm}

\noindent $\bullet$~Step A (Existence of trapping set): We claim
\[  
D^{0,\tau}(t) < \frac{1}{2}, \quad t \geq 0. 
\]
\begin{proof} 
We first estimate $D^{0,\tau}(t)$ in an interval $[-\tau, 2\tau]$. Next,
we define a set
\[
\mathcal{T} : = \left\{t \in (2\tau, \infty) : D^{0,\tau}(t)<\frac{1}{2} \right\},
\]
and proceed the proof using Lipschitz continuity of $D^{0,\tau}(t)$ in order to show $\sup\mathcal{T} = \infty$. We have 
\begin{align*}
\frac{d}{dt}D^{0,\tau}(t)&\leq \frac{\kappa_0}{8} -\frac{\kappa_0}{2} D^{0,\tau}(t)\left(2 - \frac{ D^{0,\tau}(t)^2}{2} - \frac{ D^{0,\tau}(t- \tau)^2}{2}\right)
+\frac{4\kappa_0^2\tau(N+1)}{N}\sup_{t-2\tau<v<t} D^{0,\tau}(v)\\
&< \frac{\kappa_0}{8} -\frac{7\kappa_0}{8}D^{0,\tau}(t) +\frac{\kappa_0}{18} < \frac{\kappa_0}{4}  - \frac{7\kappa_0}{8} D^{0,\tau}(t).
\end{align*}
Hence, it follows  that 
\[
D^{0,\tau}(t) \leq \max\left(D^{0,\tau}(0), \frac{\frac{\kappa_0}{4}}{\frac{7\kappa_0}{8}}\right) < \frac{1}{2}.
\]

\vspace{0.5cm}

\noindent $\bullet$~Step B (Punching step):  We claim
\[
\lim_{t\to\infty}\|z_i(s)-z_j(s)\|=0.
\]
For this, we define a Lyapunov functional $\mathcal{E}$: for $Z = (z_1, \cdots , z_N)$, 
\[
\mathcal{E}_{ij}(t) := \|z_i(t) - z_j(t)\|^2 + \gamma \int_{t-\tau}^{t}\|z_i(s) - z_j(s)\|^2 ds,
\]
where $\gamma$ is a positive constant.  Then, one has
\begin{align*}
\frac{d}{dt}\mathcal{E}_{ij}(t) &= \frac{d}{dt}\|z_i(t) - z_j(t)\|^2  + \gamma \|z_i(t) - z_j(t)\|^2  - \gamma\|z_i^\tau(t) - z_j^\tau(t)\|^2 \\
&\leq  -\frac{7\kappa_0}{4}\|z_i-z_j\|^2 + \frac{2\kappa_0}{N}\| z_i-z_j\|\cdot\|z_i^\tau-z_j^{\tau}\| \\
&+\frac{2\kappa_0}{N}\|z_i-z_j\|^2 + \gamma \|z_i(t) - z_j(t)\|^2  - \gamma \|z_i^\tau(t) - z_j^\tau(t)\|^2.
\end{align*}
By applying Young's inequality, we have 
\begin{align*}
\frac{d}{dt}\mathcal{E}_{ij}(t) &\leq  -\frac{7\kappa_0}{4}\|z_i-z_j\|^2 + \frac{\kappa_0}{N}\| z_i-z_j\|^2 + \frac{\kappa_0}{N}\|z_i^\tau-z_j^{\tau}\|^2 \\
&+\frac{2\kappa_0}{N}\|z_i-z_j\|^2 + \gamma \|z_i(t) - z_j(t)\|^2  - \gamma \|z_i^\tau(t) - z_j^\tau(t)\|^2.
\end{align*}
Now put $\gamma = \displaystyle\frac{\kappa_0}{N}$, then we get  
\[
\frac{d}{dt}\mathcal{E}_{ij}(t)\leq \left[-\frac{7\kappa_0}{4} + \frac{4\kappa_0}{N}\right]\|z_i-z_j\|^2 \leq -\frac{5\kappa_0}{12}\|z_i-z_j\|^2 \leq 0, 
\]
since $N \geq 3$.
This leads to  
\[
-\frac{5\kappa_0}{12}\int_0^\infty\|z_i(s)-z_j(s)\|^2ds\leq \mathcal{E}_{ij}(0).
\]
Using the boundedness of $\|\dot{z}_j\|$ for all $j$ we can apply Barbalat's lemma to obtain the desired result.
\end{proof}

\subsection{Close-to-SL coupling gain pair} \label{sec:3.2} In this subsection, we consider the situation in which the coupling gain pair is close to Stuart-Landau coupling gain pair:
\[
\tilde{\kappa} := \frac{\kappa_0}{2} + \kappa_1,\quad |\tilde{\kappa}| \ll 1.
\]
Note that 
\begin{align}
\begin{aligned} \label{C-7-1}
\dot{z}_j&= \displaystyle\frac{\kappa_0}{N}\sum_{k \neq j}(\langle z_j, z_j \rangle z_k^{\tau} - \langle z_k^{\tau}, z_j \rangle z_j) + \frac{\kappa_1}{N}\sum_{k \neq j}(\langle z_j, z_k^{\tau} \rangle - \langle z_k^{\tau} , z_j \rangle) z_j \\
&= \displaystyle\frac{\kappa_0}{N}\sum_{k \neq j}(\langle z_j, z_j \rangle z_k^{\tau} - \langle z_k^{\tau}, z_j \rangle z_j) +
\frac{\kappa_1}{N}\sum_{k \neq j}(\langle z_j, z_k^{\tau} \rangle - \langle z_k^{\tau} , z_j \rangle) z_j \\
&\hspace{0.3cm}-\frac{\kappa_0}{2N}\sum_{k \neq j}(\langle z_j, z_k^{\tau} \rangle - \langle z_k^{\tau} , z_j \rangle) z_j + \frac{\kappa_0}{2N}\sum_{k \neq j}(\langle z_j, z_k^{\tau} \rangle - \langle z_k^{\tau} , z_j \rangle) z_j \\
&= \displaystyle\frac{\kappa_0}{N}\sum_{k \neq j}\big( z_k^{\tau} - \mathrm{Re}( \langle z_k^{\tau}, z_j \rangle) z_j\big) +  \frac{\tilde{\kappa}}{N}\sum_{k \neq j}(\langle z_j, z_k^{\tau} \rangle - \langle z_k^{\tau} , z_j \rangle) z_j.
\end{aligned}
\end{align} 
We substitute $\kappa_1=\tilde{\kappa}-\frac{\kappa_0}{2}$ into \eqref{C-7-1} to get 
\begin{align}\label{C-8}
\begin{cases}
\dot{z}_j=\displaystyle\frac{\kappa_0}{N}\sum_{k \neq j}(\langle z_j, z_j \rangle z_k^{\tau} - \mathrm{Re}(\langle z_k^{\tau}, z_j \rangle) z_j) + \frac{\tilde{\kappa}}{N}\sum_{k \neq j}(\langle z_j, z_k^{\tau} \rangle - \langle z_k^{\tau} , z_j \rangle) z_j, \quad t > 0, \\
z_j(t) = \varphi_j(t)\in\mathbb{HS}^{d-1}, \quad -\tau \leq t \leq 0.
\end{cases}
\end{align}
By straightforward calculation, one has 
\begin{align}\label{C-9}
\begin{aligned}
&\frac{d}{dt}\|z_i - z_j^s\|^2 \\
& \hspace{0.5cm} \leq \langle \dot{z_i} - \dot{z_j}^s, z_i - z_j^s \rangle + \langle z_i - z_j^s,  \dot{z_i} - \dot{z_j}^s \rangle
= 2\mathrm{Re}\langle \dot{z_i} - \dot{z_j}^s, z_i - z_j^s \rangle \\
& \hspace{0.5cm}  =2\kappa_0\re\langle z_c^\tau-z_c^{\tau+s}, z_i-z_j^s\rangle-\kappa_0\|z_i-z_j^s\|^2(\re\langle z_c^\tau,z_i\rangle+\re\langle z_c^{\tau+s}, z_j^s\rangle)\\
& \hspace{0.5cm}  -\frac{2\kappa_0}{N}\left(\re\langle z_i-z_j^s,z_i^\tau-z_j^{\tau+s}\rangle - \|z_i-z_j^s\|^2\right)\\
& \hspace{0.5cm} + 4\tilde{\kappa}\mathrm{Im}\langle z_i, z_j^s \rangle \mathrm{Im}(\langle  z_c^{\tau} , z_i \rangle - \langle z_c^{\tau +s}, z_j^s \rangle) + \frac{4\tilde{\kappa}}{N}\mathrm{Im}\langle z_i, z_j^s \rangle \mathrm{Im}\left(\langle z_i, z_i^{\tau} \rangle - \langle z_j^s, z_j^{\tau + s} \rangle \right)\\
& \hspace{0.5cm} \leq 2\kappa_0\re\langle z_c^\tau-z_c^{\tau+s}, z_i-z_j^s\rangle-\kappa_0\|z_i-z_j^s\|^2(\re\langle z_c^\tau,z_i\rangle+\re\langle z_c^{\tau+s}, z_j^s\rangle)\\
& \hspace{0.5cm} -\frac{2\kappa_0}{N}\left(\re\langle z_i-z_j^s,z_i^\tau-z_j^{\tau+s}\rangle - \|z_i-z_j^s\|^2\right)  + 4|\tilde{\kappa}|\cdot\|z_i - z_j^s\| (\|z_c^{\tau} - z_c^{\tau +s}\| + \|z_i - z_j^s\|)\\
& \hspace{0.5cm} +\frac{4}{N}|\tilde{\kappa}|\cdot\|z_i - z_j^s\|(\|z_i - z_j^s\| + \|z_i^{\tau} - z_j^{\tau + s}\|).
\end{aligned}
\end{align} 
For the second inequality \eqref{C-9}, we use the triangle inequality and 
\begin{align*}
\|z\| = \|w \| = 1\quad \Longrightarrow\quad|\mathrm{Im}\langle z, w \rangle| \leq \|z-w\|.
\end{align*}
and similar arguments in the proof of Lemma \ref{L3.4} to derive
\begin{align}\label{C-10}
\begin{aligned}
&\Big|\|z_i(t) - z_j^s(t)\|^2 - \mathrm{Re}\langle z_i^{u}(t) - z_j^{u +s}(t), z_i(t) - z_j^s(t) \rangle\Big| \\
& \hspace{0.5cm} \leq 2u\kappa_0  \sup_{t-u < v < t}\Big(\|z_i(v) - z_j^s(v)\| + \|z_c^{\tau}(v) - z_c^{\tau +s}(v)\|\Big) \|z_i(t) -z_j^s(t)\| \\
&\hspace{0.7cm} +\frac{2u\kappa_0}{N}  \sup_{t-u < v < t} \Big(\|z_i(v) - z_j^s(v)\| + \|z_i^{\tau}(v) - z_j^{\tau +s}(v)\|\Big) \|z_i(t) -z_j^s(t)\| \\
&\hspace{0.7cm} 2u|\tilde{\kappa}|  \sup_{t-u < v < t}\Big(2\|z_i(v) - z_j^s(v)\| + \|z_c^{\tau}(v) - z_c^{\tau +s}(v)\|\Big) \|z_i(t) -z_j^s(t)\| \\
&\hspace{0.7cm} +\frac{2u|\tilde{\kappa}|}{N}  \sup_{t-u < v < t} \Big(2\|z_i(v) - z_j^s(v)\| + \|z_i^{\tau}(v) - z_j^{\tau +s}(v)\|\Big) \|z_i(t) -z_j^s(t)\|.
\end{aligned}
\end{align}

\begin{theorem} \label{T3.2}
Suppose system parameters and initial data satisfy 
\begin{align*}
\begin{aligned} 
& \kappa_0 > 0, \quad  |\tilde{\kappa}|<\frac{9}{256}\kappa_0, \quad \mathcal{C}_1\tau < \frac{1}{8}, \quad N \geq 3,  \\
& \|\varphi(t)\| = 1, \quad \sup_{-\tau \leq t \leq 0} \mathcal{D}(Z(t)) < \frac{1}{8}, \quad \Omega_j \equiv 0, \quad j \in \mathcal{N},
\end{aligned}
\end{align*}
and let $\{z_j\}$ be a global solution to \eqref{C-8}. Then complete aggregation emerges asymptotically:
\[
\lim_{t \to \infty}\|z_i(t) - z_j(t) \| = 0, \quad i,j \in \mathcal{N}.
\]
\end{theorem} 
\begin{proof} We leave its detailed proof in Section \ref{sec:3.2.2}.
\end{proof}
\subsubsection{Basic estimates} \label{sec:3.2.1} 
In this part, we provide several a priori estimates. 
\begin{lemma} \label{L3.5}
Let $\{z_j\}$ be a global solution to \eqref{C-8}. For $t \geq 2\tau$, we have following inequality:
\begin{align}\label{C-11}
\begin{aligned}
\frac{d}{dt}D^{0,\tau}(t) &\leq \kappa_0\| z_c^\tau-z_c^{2\tau}\| -\frac{\kappa_0}{2} D^{0,\tau}(t)\left(2 - \frac{ D^{0,\tau}(t)^2}{2} - \frac{ D^{0,\tau}(t- \tau)^2}{2}\right)\\
&+\frac{2\kappa_0 \tau}{N}\sup_{t-2\tau<v<t} D^{0,\tau}(v)\left(2\kappa_0\left(\frac{N+1}{N}\right)+ 3|\tilde{\kappa}|\left(\frac{N+1}{N}\right) \right) + 4|\tilde{\kappa}|\frac{N+1}{N}\sup_{t-2\tau<v<t} D^{0,\tau}(v).
\end{aligned}
\end{align}
\end{lemma} 
\begin{proof}
We set $s=\tau$ in the inequality \eqref{C-9} to find 
\begin{align*}
\begin{aligned}
\frac{d}{dt}\|z_i - z_j^{\tau}\|^2 
& \leq 2\kappa_0\re\langle z_c^\tau-z_c^{2\tau}, z_i-z_j^{\tau} \rangle-\kappa_0\|z_i-z_j^{\tau}\|^2(\re\langle z_c^\tau,z_i\rangle+\re\langle z_c^{2\tau}, z_j^{\tau}\rangle)\\
&-\frac{2\kappa_0}{N}\left(\re\langle z_i-z_j^{\tau},z_i^\tau-z_j^{2\tau}\rangle - \|z_i-z_j^{\tau}\|^2\right) + 4|\tilde{\kappa}_1|\|z_i - z_j^{\tau}\| (\|z_c^{\tau} - z_c^{2\tau}\| + \|z_i - z_j^{\tau}\|)\\
&+ \frac{4}{N}|\tilde{\kappa}|\|z_i - z_j^{\tau}\|(\|z_i - z_j^{\tau}\| + \|z_i^{\tau} - z_j^{2\tau}\|).
\end{aligned}
\end{align*}
In the inequality \eqref{C-10}, we set
\[
u=\tau\quad\text{and}\quad s = \tau
\]
to get  
\begin{align*}
\begin{aligned}
&\Big|\|z_i(t) - z_j^\tau(t)\|^2 - \mathrm{Re}\langle z_i^{\tau}(t) - z_j^{2\tau}(t), z_i(t) - z_j^\tau(t) \rangle\Big| \\
&\hspace{0.5cm}  \leq 2\tau\kappa_0  \sup_{t-\tau < v < t}\Big(\|z_i(v) - z_j^\tau(v)\| + \|z_c^{\tau}(v) - z_c^{2\tau}(v)\|\Big) \|z_i(t) -z_j^\tau(t)\| \\
&\hspace{0.5cm} +\frac{2\tau\kappa_0}{N}  \sup_{t-\tau < v < t} \Big(\|z_i(v) - z_j^\tau(v)\| + \|z_i^{\tau}(v) - z_j^{2\tau}(v)\|\Big) \|z_i(t) -z_j^\tau(t)\| \\
&\hspace{0.5cm} +2\tau|\tilde{\kappa}|  \sup_{t-\tau < v < t}\Big(2\|z_i(v) - z_j^\tau(v)\| + \|z_c^{\tau}(v) - z_c^{2\tau}(v)\|\Big) \|z_i(t) -z_j^\tau(t)\| \\
&\hspace{0.5cm} +\frac{2\tau|\tilde{\kappa}|}{N}  \sup_{t-\tau < v < t} \Big(2\|z_i(v) - z_j^\tau(v)\| + \|z_i^{\tau}(v) - z_j^{2\tau}(v)\|\Big) \|z_i(t) -z_j^\tau(t)\|.
\end{aligned}
\end{align*}
For a fixed $t$, there exist $i_t$ and $j_t$ such that 
\[
D^{0,\tau}(t) = \|z_{i_t} - z_{j_t}^\tau\|.
\]
Then, for $t \geq 2\tau$, one has 
\begin{align*}
\begin{aligned}
\frac{d}{dt}D^{0,\tau}(t)^2 &= \frac{d}{dt}\|z_{i_t} - z_{j_t}^\tau\|^2  \\
&\leq 2\kappa_0\| z_c^\tau-z_c^{2\tau}\|D^{0,\tau}(t) -\kappa_0 D^{0,\tau}(t)^2(\re\langle z_c^\tau,z_i\rangle+\re\langle z_c^{2\tau}, z_j^{\tau}\rangle)\\
&+\frac{2\kappa_0 \tau}{N}D^{0,\tau}(t) \sup_{t-2\tau<v<t} D^{0,\tau}(v)  \left(4\kappa_0+ \frac{4\kappa_0}{N} + 6|\tilde{\kappa}| + \frac{6\tilde{\kappa}}{N} \right)\\
&+ 8|\tilde{\kappa}|D^{0,\tau}(t)\sup D^{0,\tau}(t) + \frac{8|\tilde{\kappa}|}{N}D^{0,\tau}(t) \sup D^{0,\tau}(t).
\end{aligned}
\end{align*}
Hence, one has
\begin{align*}
\begin{aligned}
&\frac{d}{dt}D^{0,\tau}(t) \leq \kappa_0\| z_c^\tau-z_c^{2\tau}\| -\frac{\kappa_0}{2} D^{0,\tau}(t)(\re\langle z_c^\tau,z_i\rangle+\re\langle z_c^{2\tau}, z_j^{\tau}\rangle)\\
&\hspace{0.5cm}+\frac{\kappa_0 \tau}{N} \sup D^{0,\tau}(v) \left(4\kappa_0+ \frac{4\kappa_0}{N}+ 6|\tilde{\kappa}|  + \frac{6|\tilde{\kappa}|}{N} \right)
+ 4|\tilde{\kappa}|\sup D^{0,\tau}(t) + \frac{4|\tilde{\kappa}|}{N}\sup D^{0,\tau}(t) \\
&\leq \kappa_0\| z_c^\tau-z_c^{2\tau}\| -\frac{\kappa_0}{2} D^{0,\tau}(t)\left(2 - \frac{ D^{0,\tau}(t)^2}{2} - \frac{ D^{0,\tau}(t- \tau)^2}{2}\right)\\
&+\frac{\kappa_0 \tau}{N} \sup D^{0,\tau}(v) \left(4\kappa_0+ \frac{4\kappa_0}{N}+ 6|\tilde{\kappa}|  + \frac{6|\tilde{\kappa}|}{N} \right)
+ 4|\tilde{\kappa}|\sup_{t-2\tau<v<t} D^{0,\tau}(v) + \frac{4|\tilde{\kappa}|}{N}\sup_{t-2\tau<v<t} D^{0,\tau}(v) \\
&= \kappa_0\| z_c^\tau-z_c^{2\tau}\| -\frac{\kappa_0}{2} D^{0,\tau}(t)\left(2 - \frac{ D^{0,\tau}(t)^2}{2} - \frac{ D^{0,\tau}(t- \tau)^2}{2}\right)\\
&+\frac{2\kappa_0 \tau}{N}\sup_{t-2\tau<v<t} D^{0,\tau}(v)\left(\frac{N+1}{N}\right)\left(2\kappa_0+ 3|\tilde{\kappa}|\right)
+ 4|\tilde{\kappa}|\left(\frac{N+1}{N}\right)\sup_{t-2\tau<v<t} D^{0,\tau}(v).
\end{aligned}
\end{align*}
\end{proof}

\begin{lemma} \label{L3.6}
Let $\{z_j\}$ be a global solution to \eqref{C-8}. Then, one has 
\[
\|z_i(t) - z_i^{\tau}(t) \| \leq \mathcal{C}_1\tau
\] 
where
\[
\mathcal{C}_1=2\left(\frac{N-1}{N}\right)\cdot\min\Big \{ \kappa_0+|\kappa_1|,~\kappa_0+|\tilde{\kappa}| \Big \}.
\]
\end{lemma}

\subsubsection{Proof of Theorem \ref{T3.2}} \label{sec:3.2.2} Suppose system parameters and initial data satisfy 
\begin{align*}
\begin{aligned} 
& \kappa_0 > 0, \quad  |\tilde{\kappa}|<\frac{9}{256}\kappa_0, \quad \mathcal{C}_1\tau < \frac{1}{8}, \quad N \geq 3,  \\
& \|\varphi(t)\| = 1, \quad \sup_{-\tau \leq t \leq 0} \mathcal{D}(Z(t)) < \frac{1}{8}, \quad \Omega_j \equiv 0, \quad j \in \mathcal{N},
\end{aligned}
\end{align*}
and let $\{z_j\}$ be a solution of system \eqref{C-8}.  Then, the proof consists of two steps. \newline

\noindent $\bullet$~Step A (Existence of trapping set):~We claim:
\begin{equation} \label{C-12}
D^{0,\tau}(t) < \frac{1}{2}, \quad t \geq 0.
\end{equation}
{\it Proof of \eqref{C-12}}: We follow the same arguments as in \cite{C-H2}.  For this, we divide the estimate into three time intervals: 
\[
0 \leq t \leq \tau, \quad \tau \leq t \leq 2\tau \quad and \quad t  \geq 2\tau. 
\]
\vspace{0.2cm}

\noindent $\diamond$ Step A.1 (Estimate in the time-interval $[0,\tau]$): By triangle inequality, we have
\[
\|z_i(t) - z_j^\tau(t)\| \leq \|z_i(t) - z_i^\tau(t)\| + \|z_i^\tau(t) - z_j^\tau(t)\| \leq \mathcal{C}_1\tau + \mathcal{D}(Z(t-\tau)) < \frac{1}{4}.
\]

\vspace{0.2cm}

\noindent $\diamond$ Step A.2 (Estimate in the time-interval $[\tau, 2\tau]$): Similar to Step A,  we use triangular inequality to get
\[
\|z_i(t) - z_j^\tau(t)\| \leq \|z_i(t) - z_i^\tau(t)\| + \|z_i^\tau(t) - z_j^\tau(t)\| \leq \mathcal{C}_1\tau + \mathcal{D}(Z(t-\tau)).
\]
However, since 
\[
\|z_i(t) - z_j(t) \| \leq \|z_i(t) - z_i^\tau(t) \| + \|z_i^\tau(t) - z_j(t) \| \leq \mathcal{C}_1\tau + \frac{1}{4}<\frac{3}{8},
\]
one has
\[  \mathcal{D}(Z(t-\tau)) < \frac{3}{8}. \]
Therefore, we give
\[ \|z_i(t) - z_j^\tau(t)\| <\frac{1}{2}. \] 

\vspace{0.2cm}

\noindent $\diamond$ Step A.3 (Estimate in the time-interval $[2\tau, \infty)$):~ By \eqref{C-11}, one has 
\begin{align*}
\frac{d}{dt}D^{0,\tau}(t) &\leq \kappa_0\| z_c^\tau-z_c^{2\tau}\| -\frac{\kappa_0}{2} D^{0,\tau}(t)\left(2 - \frac{ D^{0,\tau}(t)^2}{2} - \frac{ D^{0,\tau}(t- \tau)^2}{2}\right)\\
&\hspace{0.5cm} +\frac{2\kappa_0 \tau}{N}\sup_{t-2\tau<v<t} D^{0,\tau}(v)\left(\frac{N+1}{N}\right)\left(2\kappa_0+ 3|\tilde{\kappa}| \right)
+ 4|\tilde{\kappa}|\left(\frac{N+1}{N}\right)\sup_{t-2\tau<v<t} D^{0,\tau}(v)\\
&\leq \frac{\kappa_0}{8} -\frac{\kappa_0}{2} D^{0,\tau}(t)\left(2 - \frac{ D^{0,\tau}(t)^2}{2} - \frac{ D^{0,\tau}(t- \tau)^2}{2}\right) +\frac{3\kappa_0}{4N}\sup_{t-2\tau<v<t} D^{0,\tau}(v) \\
&\hspace{0.5cm}+  4|\tilde{\kappa}|\frac{N+1}{N}\sup_{t-2\tau<v<t} D^{0,\tau}(v).
\end{align*}
Here, we used 
\[
\| z_c^\tau-z_c^{2\tau}\| \leq \mathcal{C}_1\tau < \frac{1}{8}, \quad
\left(\frac{N+1}{N}\right)\left(4\kappa_0 + 6|\tilde{\kappa}| \right) <\frac{6(N+1)}{N}(\kappa_0+|\tilde{\kappa}|)\leq \frac{3(N+1)}{N-1}\mathcal{C}_1\leq 6\mathcal{C}_1
\]
when $N\geq3$. \newline

\noindent Next we claim:
\[
D^{0,\tau}(t) < \frac{1}{2}, \quad \forall ~~ t \geq 2\tau.
\]
For the proof, we define a set $\mathcal{T}$ as 
\[
\mathcal{T} : = \left\{t \in (2\tau, \infty) : D^{0,\tau}(t)<\frac{1}{2} \right\},
\]
and proceed the proof using Lipschitz continuity of $D^{0,\tau}(t)$ as in \cite{C-H2}. The only difference is the estimate of $\frac{d}{dt}D^{0,\tau}(t)$. By direct estimates, one has 
\begin{align*}
\frac{d}{dt}D^{0,\tau}(t)&\leq \frac{\kappa_0}{8} -\frac{\kappa_0}{2} D^{0,\tau}(t)\left(2 - \frac{ D^{0,\tau}(t)^2}{2} - \frac{ D^{0,\tau}(t- \tau)^2}{2}\right)\\
&\hspace{0.5cm} +\frac{3\kappa_0}{4N}\sup_{t-2\tau<v<t} D^{0,\tau}(v)+ \sup_{t-2\tau<v<t} D^{0,\tau}(v)\left(4|\tilde{\kappa}|+ \frac{4|\tilde{\kappa}|}{N}\right)\\
&< \frac{\kappa_0}{8} -\frac{7\kappa_0}{8}D^{0,\tau}(t) +\frac{\kappa_0}{8}+ \frac{8}{3}|\tilde{\kappa} =\left(\frac{\kappa_0}{4}+\frac{8}{3}|\tilde{\kappa}|\right)   -\frac{7\kappa_0}{8} D^{0,\tau}(t).
\end{align*}
Hence, it follows from $|\tilde{\kappa}|<\frac{9}{256}\kappa_0$ that 
\[
\frac{\frac{\kappa_0}{4}+\frac{8}{3}|\tilde{\kappa}|}{\frac{7\kappa_0}{8}}<\frac{1}{2}, \quad 
D^{0,\tau}(t) \leq \max\Big \{ D^{0,\tau}(0),~~\frac{\frac{\kappa_0}{4}+\frac{8}{3}|\tilde{\kappa}|}{\frac{7\kappa_0}{8}} \Big \} < \frac{1}{2}.
\]
In this way, we verified claim \eqref{C-12}. 

\vspace{0.5cm}

\noindent $\bullet$~Step B (Zero convergence of modified diameter):~We claim
\begin{equation} \label{C-13}
\lim_{t\to\infty}\|z_i(s)-z_j(s)\|=0.
\end{equation}
The proof is similar to Theorem 3.1 of \cite{C-H2} with a slight difference. We present main steps that involve such differences. We put $s= 0$ in \eqref{C-9} to get 
\begin{align*}
\frac{d}{dt}\|z_i - z_j\|^2  &\leq -\kappa_0\|z_i-z_j\|^2(\re\langle z_c^\tau,z_i\rangle+\re\langle z_c^{\tau}, z_j\rangle)\\
&\hspace{0.2cm}-\frac{2\kappa_0}{N}\left(\re\langle z_i-z_j,z_i^\tau-z_j^{\tau}\rangle - \|z_i-z_j\|^2\right) + 4|\tilde{\kappa}|\cdot\|z_i - z_j\|^2\\
&\hspace{0.2cm} + \frac{4}{N}|\tilde{\kappa}|\cdot\|z_i - z_j\|(\|z_i - z_j\| + \|z_i^{\tau} - z_j^{\tau}\|).
\end{align*} 
Next, we define a Lyapunov functional $\mathcal{E}_{ij}$ for $Z = (z_1, \cdots , z_N)$ and $i, j\in\mathcal{N}$:
\[
\mathcal{E}_{ij}(t) := \|z_i(t) - z_j(t)\|^2 + \gamma \int_{t-\tau}^{t}\|z_i(s) - z_j(s)\|^2 ds,
\]
where $\gamma$ is a positive constant. Then, one has
\begin{align*}
\frac{d}{dt}\mathcal{E}_{ij}(t) &= \frac{d}{dt}\|z_i(t) - z_j(t)\|^2  + \gamma \|z_i(t) - z_j(t)\|^2  - \gamma\|z_i^\tau(t) - z_j^\tau(t)\|^2 \\
&\leq  -\frac{7\kappa_0}{4}\|z_i-z_j\|^2 + \frac{2\kappa_0}{N}\| z_i-z_j\|\cdot\|z_i^\tau-z_j^{\tau}\| \\
&\hspace{0.5cm}+\frac{2\kappa_0}{N}\|z_i-z_j\|^2 + 4|\tilde{\kappa}|\cdot\|z_i - z_j\| (\|z_c^{\tau} - z_c^{\tau}\| + \|z_i - z_j\|)\\
&\hspace{0.5cm}+ \frac{4}{N}|\tilde{\kappa}|\cdot\|z_i - z_j\|(\|z_i - z_j\| + \|z_i^{\tau} - z_j^{\tau}\|) + \gamma \|z_i(t) - z_j(t)\|^2  - \gamma \|z_i^\tau(t) - z_j^\tau(t)\|^2.
\end{align*}
By Young's inequality, we have
\begin{align*}
\begin{aligned}
\frac{d}{dt}\mathcal{E}_{ij}(t)&\leq -\frac{7\kappa_0}{4} \|z_i-z_j\|^2 + \frac{\kappa_0}{N\alpha}\|z_i-z_j\|^2 + \frac{\kappa_0 \alpha}{N}\|z_i^\tau-z_j^\tau \|^2 + \frac{2\kappa_0}{N}\|z_i-z_j\|^2\\
&\hspace{0.5cm}+ 4|\tilde{\kappa}|\|z_i - z_j\|^2 + 4\frac{|\tilde{\kappa}|}{N}\|z_i - z_j\|^2 + 2\frac{|\tilde{\kappa}|}{N}\left(\frac{1}{\alpha}\|z_i - z_j\|^2 + \alpha \|z_i^\tau - z_j^\tau \|^2\right)  \\
&\hspace{0.5cm}+ \gamma \|z_i(t) - z_j(t)\|^2  - \gamma \|z_i^\tau(t) - z_j^\tau(t)\|^2 \\
&=\left(-\frac{7}{4}\kappa_0 + \frac{\kappa_0}{N\alpha} + \frac{2\kappa_0}{N} + 4|\tilde{\kappa}| +\frac{4|\tilde{\kappa}|}{N} +  \frac{2|\tilde{\kappa}|}{N\alpha}+\gamma \right) \|z_i - z_j\|^2 \\
&\hspace{0.5cm}+ \left(\frac{\kappa_0 \alpha}{N} - \gamma + \frac{2\alpha |\tilde{\kappa}|}{N}\right) \|z_i^\tau - z_j^\tau \|^2.
\end{aligned}
\end{align*}
Now, we set
\[
\gamma = \frac{\kappa_0 \alpha}{N} + \frac{2\alpha |\tilde{\kappa}|}{N} \quad \mbox{and} \quad \alpha = 1.
\]
Then, we have
\begin{align*}
\frac{d}{dt}\mathcal{E}_{ij}(t) &\leq \left(-\frac{7}{4}\kappa_0 + \frac{\kappa_0}{N\alpha} + \frac{2\kappa_0}{N} + 4|\tilde{\kappa}| +\frac{4|\tilde{\kappa}|}{N} +  \frac{2|\tilde{\kappa}|}{N\alpha}+\frac{\kappa_0 \alpha}{N} + \frac{2\alpha |\tilde{\kappa}|}{N}\right) \|z_i - z_j\|^2 \\
&\leq \left(-\frac{7}{4}\kappa_0 + \frac{4\kappa_0}{N} + 4|\tilde{\kappa}| +\frac{8|\tilde{\kappa}|}{N}\right) \|z_i - z_j\|^2.
\end{align*}
For $N \geq 3$ and $|\tilde{\kappa}| < \frac{1}{16}\kappa_0$, we have 
\[
-\frac{7}{4}\kappa_0 + \frac{4\kappa_0}{N} + 4|\tilde{\kappa}| +\frac{8|\tilde{\kappa}|}{N} < 0.
\]
Here we set
\[
\beta=-\left(-\frac{7}{4}\kappa_0 + \frac{4\kappa_0}{N} + 4|\tilde{\kappa}| +\frac{8|\tilde{\kappa}|}{N}\right)
\]
to obtain
\[
\frac{d}{dt}\mathcal{E}_{ij}(t)\leq -\beta\|z_i-z_j\|^2.
\]
This yields
\[
\mathcal{E}_{ij}(t)-\mathcal{E}_{ij}(0)\leq -\beta\int_0^t\|z_i(s)-z_j(s)\|^2ds
\]
which is equivalent to 
\[
\mathcal{E}_{ij}(t)+\beta\int_0^t\|z_i(s)-z_j(s)\|^2ds\leq \mathcal{E}_{ij}(0).
\]
It follows from definition of $\mathcal{E}_{ij}$ that 
\[ \mathcal{E}_{ij}\geq0. \]
Finally, we have
\[
\beta\int_0^t\|z_i(s)-z_j(s)\|^2ds\leq \mathcal{E}_{ij}(0).
\]
By letting $t\to\infty$, one has 
\[
\beta\int_0^\infty\|z_i(s)-z_j(s)\|^2ds\leq \mathcal{E}_{ij}(0).
\]
It follows from the boundedness of $\|\dot{z}_j\|$ for all $j$ that 
\[ \sup_{0 \leq t < \infty} \Big| \frac{d}{dt}\|z_i(s)-z_j(s)\|^2 \Big| < \infty. \]
 This means $\|z_i(s)-z_j(s)\|$ is uniformly continuous. Hence, we can apply Barbalat's lemma to obtain the desired estimate \eqref{C-13}.
\section{Emergence of practical aggregation} \label{sec:4}
\setcounter{equation}{0}
In this section, we present practical aggregation of the LHS model.
\subsection{Complete network topology} \label{sec:4.1} In this subsection, we set 
\[ a_{ij}\equiv 1, \quad \Omega_j\equiv0, \quad i, j \in {\mathcal N}. \]
In this case, system \eqref{A-1} becomes 
\begin{equation} \label{D-1}
\begin{cases} 
\displaystyle \dot{z}_j = \frac{\kappa_0}{N}\sum_{k \neq j} \big(\langle z_j, z_j \rangle z_k^{\tau} - \langle z_k^{\tau}, z_j \rangle z_j) + \frac{\kappa_1}{N}\sum_{k \neq j} \big(\langle z_j, z_k^{\tau} \rangle - \langle z_k^{\tau} , z_j \rangle) z_j, \quad t > 0, \vspace{0.2cm} \\
\displaystyle z_j(t) = \varphi_j(t)\in\bbc^d, \quad -\tau \leq t \leq 0,~j \in {\mathcal N}.
\end{cases}
\end{equation}
For handy notation, we define the following notation:
\begin{equation} \label{D-2}
G_{ij} := \langle z_i, z_j \rangle ,\qquad G_{ij}^{\tau} := \langle z_i^{\tau}, z_j \rangle,\qquad L_{ij}=1-G_{ij},\quad L_{ij}^\tau=1-G_{ij}^\tau.
\end{equation}
Our third main result is concerned with the practical aggregation. Recall that 
\[  L(t) = \displaystyle\max_{i,j}|1-\langle z_i(t), z_j(t) \rangle|. \]
\begin{theorem} \label{T4.1}
Suppose coupling gains and initial data satisfy 
\[  2|\kappa_1|<\kappa_0, \quad L(0)<1 = \frac{2|\kappa_1|}{\kappa_0}, \]
and let $\{z_j\}$ be a global solution to \eqref{C-8}. Then, system \eqref{C-8} exhibits the practical synchronization:
\[
\lim_{\tau\searrow0}\limsup_{t\to\infty}L(t) =0.
\]
\end{theorem}
\begin{proof}
We leave its proof in Section \ref{sec:4.1.2}.
\end{proof}

\subsubsection{Basic estimates} \label{sec:4.1.1}
In this part, we provide several lemmas to be crucially used in the proof of Theorem \ref{T4.1}.
\begin{lemma} \label{L4.1}
Let $\{z_j\}$ be a global solution to \eqref{D-2}. Then, $G_{ij}$ satisfies 
\begin{align*}
\begin{aligned}
\frac{d}{dt} G_{ij} &= \displaystyle\frac{\kappa_0}{N}\sum_{k \neq i}(G_{kj}^{\tau}-\overline{G}_{ki}^{\tau}G_{ij}) + \displaystyle\frac{\kappa_0}{N}\sum_{k \neq j}(\overline{G}_{ki}^{\tau} - G_{kj}^{\tau}G_{ij})
 \\
&\hspace{0.5cm} + \displaystyle\frac{\kappa_1}{N}\sum_{k \neq i}(G_{ki}^{\tau}-\overline{G}_{ki}^{\tau})G_{ij}+ \displaystyle\frac{\kappa_1}{N}\sum_{k \neq j}(\overline{G}_{kj}^{\tau} - G_{kj}^{\tau})G_{ij}.
\end{aligned}
\end{align*}
\end{lemma}
\begin{proof} By direct calculation, one has 
\begin{align*}
\begin{aligned}
\frac{d}{dt} G_{ij} &=\langle \dot{z}_i, z_j\rangle+\langle z_i, \dot{z}_j\rangle\\
&= \left\langle\displaystyle\frac{\kappa_0}{N}\sum_{k \neq i}\big(\langle z_i, z_i \rangle z_k^{\tau} - \langle z_k^{\tau}, z_i \rangle z_i) + \frac{\kappa_1}{N}\sum_{k \neq i}\big(\langle z_i, z_k^{\tau} \rangle - \langle z_k^{\tau} , z_i \rangle) z_i, z_j \right\rangle\\
&\hspace{0.5cm} + \left\langle z_i,  \displaystyle\frac{\kappa_0}{N}\sum_{k \neq j}\big(\langle z_j, z_j \rangle z_k^{\tau} - \langle z_k^{\tau}, z_j \rangle z_j) + \frac{\kappa_1}{N}\sum_{k \neq j}\big(\langle z_j, z_k^{\tau} \rangle - \langle z_k^{\tau} , z_j \rangle) z_j \right\rangle\\
&=  \displaystyle\frac{\kappa_0}{N}\sum_{k \neq i}\big(\langle z_k^{\tau}, z_j \rangle - \overline{\langle z_k^{\tau}, z_i \rangle} \langle z_i, z_j \rangle) + \frac{\kappa_1}{N}\sum_{k \neq i}\big(\overline{\langle z_i, z_k^{\tau} \rangle} - \overline{\langle z_k^{\tau} , z_i \rangle}) \langle z_i, z_j \rangle\\
&\hspace{0.5cm}  + \displaystyle\frac{\kappa_0}{N}\sum_{k \neq j}\big( \langle z_i, z_k^{\tau} \rangle - \langle z_k^{\tau}, z_j \rangle \langle z_i, z_j \rangle) + \frac{\kappa_1}{N}\sum_{k \neq j}\big(\langle z_j, z_k^{\tau} \rangle  - \langle z_k^{\tau} , z_j \rangle) \langle z_i, z_j \rangle \\
&=\displaystyle\frac{\kappa_0}{N}\sum_{k \neq i}(G_{kj}^{\tau}-\overline{G}_{ki}^{\tau}G_{ij}) + \displaystyle\frac{\kappa_0}{N}\sum_{k \neq j}(\overline{G}_{ki}^{\tau} - G_{kj}^{\tau}G_{ij})
+ \displaystyle\frac{\kappa_1}{N}\sum_{k \neq i}(G_{ki}^{\tau}-\overline{G}_{ki}^{\tau})G_{ij} \\
&\hspace{0.5cm} + \frac{\kappa_1}{N}\sum_{k \neq j}(\overline{G}_{kj}^{\tau} - G_{kj}^{\tau})G_{ij}.
\end{aligned}
\end{align*}
\end{proof}
\begin{lemma} \label{L4.2}
Let $A\in\bbc^{d\times d}$ and $v\in\bbc^{d}$ be given matrix and vector, respectively. Then, one has
\[
\|Av\|\leq \|A\|_F\cdot\|v\|,
\]
where $\|\cdot\|$ is a vector norm in $\bbc^d$ and $\|\cdot\|_F$ is a Frobenius norm.
\end{lemma}
\begin{proof}
We set the componentwise form of $A$ and $v$ as follows: 
\[  A := [A]_{\alpha\beta} \quad \mbox{and} \quad  v:=[v]_\gamma, \]
where $1\leq \alpha,\beta,\gamma\leq d$. By the Cauchy-Schwarz inequality, we have
\[
|[Av]_\alpha|=\left|\sum_{\beta=1}^{d}[A]_{\alpha\beta}[v]_\beta\right| \leq \sqrt{\sum_{\beta=1}^{d}[\bar{A}]_{\alpha\beta}[A]_{\alpha\beta}}\cdot\sqrt{\sum_{\beta=1}^{d}[\bar{v}]_\beta [v]_\beta},
\]
Thus, one has
\begin{align*}
\|Av\|^2=\sum_{\alpha=1}^d|[Av]_\alpha|^2\leq \left(\sum_{\beta=1}^d[\bar{A}]_{\alpha\beta}[A]_{\alpha\beta}\right)\cdot\left(\sum_{\beta=1}^d[\bar{v}]_\beta [v]_\beta\right)=\|A\|_F^2\cdot\|v\|^2,
\end{align*}
and this yields the desired result.
\end{proof}
\begin{lemma} \label{L4.3}
Let $\{z_j\}$ be a global solution to \eqref{D-1}. Then $L_{ij}$ in \eqref{D-2} satisfies 
\[ 
| L_{ij}(t)-L_{ij}^{\tau}(t)| \leq \tau \mathcal{C}_2,
\]
where the positive constant $\mathcal{C}_2$ is given by
\[
\mathcal{C}_2 = \frac{2(N-1)}{N}\left(\kappa_0 + |\kappa_1|\right).
\]
\end{lemma}
\begin{proof}
By Cauchy-Schwarz inequality, we have
\[
|L_{ij}(t)-L_{ij}^{\tau}(t)| = |\langle z_i - z_i^{\tau} , z_j \rangle| \leq \|z_i - z_i^{\tau} \| \cdot\|z_j \|.
\]
Note that $\|z_j\| = 1$. By Lemma \eqref{L3.6}, we have
\[ 
\|z_i - z_i^{\tau} \| \leq  \tau \Big\{  2 \frac{N-1}{N}(\kappa_0 + |\kappa_1|) \Big \}.
\]
\end{proof}
\begin{lemma}\label{L4.4}
Let $\{z_j\}$ be a global solution to \eqref{D-1}. Then, $|L_{ij}|$ satisfies 
\begin{align*}
\begin{aligned}
\frac{d}{dt}|L_{ij}|^2  &\leq -\frac{2\kappa_0}{N}\sum_{k=1}^{N}|L_{ij}|^2(\mathrm{Re}( \langle z_k^{\tau}, z_i + z_j \rangle) + 4|\kappa_1|\cdot |L_{ij}|\cdot(|L_{ci}|+|L_{cj}|+ 2\mathcal{C}_2\tau)\\
&\hspace{0.5cm} +|L_{ij}|\frac{8\mathcal{C}_2\tau}{N}(\kappa_0 + |\kappa_1|) + |L_{ij}|^2\frac{4\kappa_0 \mathcal{C}_2\tau}{N}.
\end{aligned}
\end{align*}
\end{lemma}
\begin{proof}
We use \eqref{D-1} to get 
\begin{align*}
\begin{aligned}
\frac{d}{dt} G_{ij} &= \displaystyle\frac{\kappa_0}{N}\sum_{k \neq i}(G_{kj}^{\tau}-\overline{G}_{ki}^{\tau}G_{ij}) + \displaystyle\frac{\kappa_0}{N}\sum_{k \neq j}(\overline{G}_{ki}^{\tau} - G_{kj}^{\tau}G_{ij})
+ \displaystyle\frac{\kappa_1}{N}\sum_{k \neq i}(G_{ki}^{\tau}G_{ij}-\overline{G}_{ki}^{\tau}G_{ij}) \\
&\hspace{0.5cm} + \displaystyle\frac{\kappa_1}{N}\sum_{k \neq j}(\overline{G}_{kj}^{\tau}G_{ij} - G_{kj}^{\tau}G_{ij})\\
&= \displaystyle\frac{\kappa_0}{N}\sum_{k = 1}^{N}(G_{kj}^{\tau}-\overline{G}_{ki}^{\tau}G_{ij} + \overline{G}_{ki}^{\tau} - G_{kj}^{\tau}G_{ij}) + \displaystyle\frac{\kappa_1}{N}\sum_{k =1}^{N}(G_{ki}^{\tau}G_{ij}-\overline{G}_{ki}^{\tau}G_{ij}+\overline{G}_{kj}^{\tau}G_{ij} - G_{kj}^{\tau}G_{ij}) \\
&\hspace{0.5cm} - \frac{\kappa_0}{N}(G_{ij}^{\tau}-\overline{G}_{ii}^{\tau}G_{ij} + \overline{G}_{ji}^{\tau}-G_{jj}^{\tau}G_{ij}) - \frac{\kappa_1}{N}(G_{ii}^{\tau}G_{ij}-\overline{G}_{ii}^{\tau}G_{ij} + \overline{G}_{jj}^{\tau}G_{ij}-G_{jj}^{\tau}G_{ij})\\
&= \displaystyle\frac{\kappa_0}{N}\sum_{k = 1}^{N}(2-L_{kj}^{\tau}-\overline{L}_{ki}^{\tau})L_{ij} + \displaystyle\frac{2\mathrm{i}\kappa_1}{N}\sum_{k =1}^{N}(\mathrm{Im}L_{kj}^{\tau}-\mathrm{Im}L_{ki}^{\tau})(1-L_{ij}) \\
&\hspace{0.5cm} - \frac{\kappa_0}{N}(2L_{ij}-L_{ij}^{\tau}-\overline{L}_{ji}^{\tau}+L_{jj}^{\tau}+\overline{L}_{ii}^{\tau}-\overline{L}_{ii}^{\tau}L_{ij}-L_{jj}^{\tau}L_{ij}) - \frac{2\mathrm{i}\kappa_1}{N}(1-L_{ij})(\mathrm{Im}L_{jj}^{\tau}-\mathrm{Im}L_{ii}^{\tau})
\end{aligned}
\end{align*}
Thus, we have
\begin{align*}
\frac{d}{dt}|L_{ij}|^2 &= \frac{d}{dt}(L_{ij} \bar{L}_{ij}) = \dot{L}_{ij}\overline{L}_{ij} + L_{ij}\dot{\overline{L}}_{ij} = -\frac{d}{dt}\langle z_i, z_j \rangle (1 - \langle z_j, z_i \rangle) + (1- \langle z_i, z_j \rangle)(-\frac{d}{dt}\langle z_j, z_i \rangle)\\
&= -L_{ji}\big(\displaystyle\frac{\kappa_0}{N}\sum_{k = 1}^{N}(2-L_{kj}^{\tau}-\overline{L}_{ki}^{\tau})L_{ij} + \displaystyle\frac{2\mathrm{i} \kappa_1}{N}\sum_{k =1}^{N}(\mathrm{Im}L_{kj}^{\tau}-\mathrm{Im}L_{ki}^{\tau})(1-L_{ij})\\
&\hspace{0.5cm} - \frac{\kappa_0}{N}(2L_{ij}-L_{ij}^{\tau}-\overline{L}_{ji}^{\tau}+L_{jj}^{\tau}+\overline{L}_{ii}^{\tau}-\overline{L}_{ii}^{\tau}L_{ij}-L_{jj}^{\tau}L_{ij}) - \frac{2\mathrm{i}\kappa_1}{N}(1-L_{ij})(\mathrm{Im}L_{jj}^{\tau}-\mathrm{Im}L_{ii}^{\tau})\\
&\hspace{0.5cm} -L_{ij}\big(\displaystyle\frac{\kappa_0}{N}\sum_{k = 1}^{N}(2-L_{ki}^{\tau}-\overline{L}_{kj}^{\tau})L_{ji} + \displaystyle\frac{2\mathrm{i}\kappa_1}{N}\sum_{k =1}^{N}(\mathrm{Im}L_{ki}^{\tau}-\mathrm{Im}L_{kj}^{\tau})(1-L_{ji}) \\
&\hspace{0.5cm} - \frac{\kappa_0}{N}(2L_{ji}-L_{ji}^{\tau}-\overline{L}_{ij}^{\tau}+L_{ii}^{\tau}+\overline{L}_{jj}^{\tau}-\overline{L}_{jj}^{\tau}L_{ji}-L_{ii}^{\tau}L_{ji}) - \frac{2\mathrm{i}\kappa_1}{N}(1-L_{ji})(\mathrm{Im}L_{ii}^{\tau}-\mathrm{Im}L_{jj}^{\tau}))\\
&= \frac{\kappa_0}{N}\sum_{k=1}^{N}L_{ij}L_{ji}(L_{kj}^{\tau}+\overline{L}_{kj}^{\tau}+L_{ki}^{\tau}+\overline{L}_{ki}^{\tau}-4)
+ \frac{2\mathrm{i}\kappa_1}{N}\sum_{k=1}^{N}(L_{ij}-L_{ji})(\mathrm{Im}L_{kj}^{\tau}-\mathrm{Im}L_{ki}^{\tau})\\
&\hspace{0.5cm} + \frac{\kappa_0}{N}\big(4L_{ij}L_{ji} + (-L_{ij}^{\tau}-\overline{L}_{ji}^{\tau}+L_{jj}^{\tau}+\overline{L}_{ii}^{\tau})L_{ji} +(-L_{ji}^{\tau}-\overline{L}_{ij}^{\tau}+L_{ii}^{\tau}+\overline{L}_{jj}^{\tau})L_{ij} \\
&\hspace{0.5cm} -(\overline{L}_{ii}^{\tau}+L_{jj}^{\tau}+\overline{L}_{jj}^{\tau}+L_{ii}^{\tau})L_{ij}L_{ji}) + \frac{2\mathrm{i}\kappa_1}{N}(L_{ij}-L_{ji})(\mathrm{Im}L_{ii}^{\tau} - \mathrm{Im}L_{jj}^{\tau}).
\end{align*}
Note that 
\[ L_{ij}L_{ji} = (1- \langle z_i, z_j \rangle)(1- \langle z_j, z_i \rangle) = \| L_{ij}\| ^2  \quad \mbox{and} \quad \overline{L_{ij}}=L_{ji}. \]
So we have
\begin{align*}
\frac{d}{dt}|L_{ij}|^2 &= \frac{2\kappa_0}{N}\sum_{k=1}^{N}|L_{ij}|^2(\mathrm{Re}L_{ki}^{\tau} + \mathrm{Re}L_{kj}^{\tau} - 2) + \frac{4\kappa_1}{N}\sum_{k = 1}^{N}\mathrm{Im}L_{ij}(\mathrm{Im}L_{ki}^{\tau} - \mathrm{Im}L_{kj}^{\tau})\\
&\hspace{0.5cm} + \frac{\kappa_0}{N}\big(4|L_{ij}|^2 + 2\mathrm{Re}((-L_{ji}^{\tau} - \overline{L}_{ij}^{\tau} + L_{ii}^{\tau} + \overline{L}_{jj}^{\tau})L_{ij}) -2(\mathrm{Re}(L_{ii}^{\tau} + L_{jj}^{\tau})|Lij|^2)\big) \\
&\hspace{0.5cm} - \frac{4\kappa_1}{N}\mathrm{Im}L_{ij}\mathrm{Im}(L_{ii}^{\tau} - L_{jj}^{\tau}).
\end{align*}
Note that the last two terms of the right hand side in above equation goes to zero as $\tau$ goes to $0$. \newline

\noindent $\bullet$~Step A: Note that 
\begin{align*}
\begin{aligned}
& \|z_i - z_j\|^2 = |2(1 - \re(\langle z_i, z_j \rangle)| \leq 2|L_{ij}|, \\
\end{aligned}
\end{align*}
By Lemma 4.3, we have for any $i$,
\[
|L_{ii}^{\tau}| \leq \tau \mathcal{C}_2.
\]
Since $L_{ii} = 0$, one has 
\begin{align*}
\begin{aligned}
\left|\frac{4\kappa_1}{N}\mathrm{Im}L_{ij}\mathrm{Im}(L_{ii}^{\tau} - L_{jj}^{\tau}) \right| &= \frac{4|\kappa_1|}{N}|L_{ij}| \mathrm{Im}(\langle z_i, z_i^{\tau} \rangle - \langle z_j, z_j^{\tau} \rangle)| \\
& =  \frac{4|\kappa_1|}{N}|L_{ij}| \mathrm{Im}(\langle z_i - z_j, z_i^{\tau} \rangle + \langle z_j, z_i^{\tau} - z_j^{\tau} \rangle)| \\
&  \leq \frac{4|\kappa_1|}{N}|L_{ij}|(|L_{ii}^{\tau}| + |L_{jj}^{\tau}|) \leq \frac{16\mathcal{C}_2|\kappa_1|\tau}{N}|L_{ij}|.
\end{aligned}
\end{align*}

\vspace{0.2cm}

\noindent $\bullet$ Step B: Next, we analyze the term
\begin{align}
\begin{aligned} \label{D-3}
A &:= \frac{\kappa_0}{N}\big(4|L_{ij}|^2 + 2\mathrm{Re}((-L_{ji}^{\tau} - \overline{L}_{ij}^{\tau} + L_{ii}^{\tau} + \overline{L}_{jj}^{\tau})L_{ij}) -2(\mathrm{Re}(L_{ii}^{\tau} + L_{jj}^{\tau})|L_{ij}|^2) \\
 &= \frac{\kappa_0}{N}\Big(  \underbrace{4|L_{ij}|^2 + 2\mathrm{Re}((-L_{ji}^{\tau}L_{ij} - \overline{L}_{ij}^{\tau}L_{ij})}_{=: A_1} + \underbrace{2\mathrm{Re}(L_{ii}^{\tau}L_{ij} + \overline{L}_{jj}^{\tau}L_{ij}) -2(\mathrm{Re}(L_{ii}^{\tau} + L_{jj}^{\tau})|L_{ij}|^2}_{=: A_2} \Big).
\end{aligned}
\end{align}
In the sequel, we estimate $A_i,~i = 1,2$ as follows. \newline

\noindent $\bullet$~(Estimate of $A_2$):~By direct estimate, one has 
\begin{align} 
\begin{aligned} \label{D-4}
A_2 &\leq 2|L_{ii}^{\tau}L_{ij}| + 2| \overline{L}_{jj}^{\tau}L_{ij}|  + 2\big((|L_{ii}^{\tau}| + |L_{jj}^{\tau}|)|L_{ij}|^2\big) \\
&\leq 2|L_{ij}|(|L_{ii}^{\tau}| + | \overline{L}_{jj}^{\tau}|)  + 2|L_{ij}|^2(|L_{ii}^{\tau}| + |L_{jj}^{\tau}|)\\
&\leq 4\mathcal{C}_2\tau|L_{ij}| + 4\mathcal{C}_2\tau|L_{ij}|^2
\end{aligned}
\end{align}
On the other hand, note that
\begin{align*}
|L_{ij}|^2 - \mathrm{Re}L_{ji}^{\tau}L_{ij} &= \mathrm{Re}|L_{ij}|^2 - \mathrm{Re}L_{ji}^{\tau}L_{ij} = \mathrm{Re}\big((L_{ji}-L_{ji}^{\tau})L_{ij}\big) \\
&\leq |\big((L_{ji}-L_{ji}^{\tau})L_{ij}\big)| = |L_{ji}-L_{ji}^{\tau}||L_{ij}| \leq \mathcal{C}_2\tau|L_{ij}|
\end{align*}

\vspace{0.2cm}

\noindent $\bullet$~(Estimate of $A_1$):~Similarly, one has 
\begin{align*}
\begin{aligned} 
|L_{ij}|^2 - \mathrm{Re}L_{ij}\overline{L}_{ij}^{\tau} &= \mathrm{Re}(|L_{ij}|^2 - L_{ij}\overline{L}_{ij}^{\tau}) = \mathrm{Re}(L_{ij}L_{ji} - L_{ij}\overline{L}_{ij}^{\tau})\\
&\leq |L_{ij}||L_{ji} - \overline{L}_{ij}^{\tau}| = |L_{ij}|\cdot| 1 - \langle z_j, z_i \rangle - 1 + \overline{\langle z_i^{\tau}, z_j \rangle}| \\
&\leq |L_{ij}| |- \langle z_j, z_i \rangle + \langle z_j, z_i^{\tau}\rangle|=|L_{ij}|\cdot\|\langle z_j, z_i - z_i^{\tau} \rangle \| \leq \mathcal{C}_2\tau |L_{ij}|.
\end{aligned}
\end{align*}
Thus, we have
\begin{align}
\begin{aligned} \label{D-5}
A_1 &= 4\|L_{ij}\|^2 + 2\mathrm{Re}((-L_{ji}^{\tau}L_{ij} - \overline{L}_{ij}^{\tau}L_{ij})  \\
&= 2(\|L_{ij}\|^2 + \|L_{ij}\|^2 - \mathrm{Re}L_{ji}^{\tau}L_{ij} - \mathrm{Re}\overline{L}_{ij}^{\tau}L_{ij}) \\
&\leq 4\mathcal{C}_2\tau\|L_{ij}\|.
\end{aligned}
\end{align}
In \eqref{D-3}, we combine all the estimate \eqref{D-4} and \eqref{D-5} to find 
\[
A = \frac{\kappa_0}{N}(A_1 + A_2) \leq \frac{4\kappa_0 \mathcal{C}_2\tau}{N}(2|L_{ij}| + |L_{ij}|^2), 
\]
and
\begin{align*}
\begin{aligned}
&\frac{\kappa_0}{N}\big(4|L_{ij}|^2 + 2\mathrm{Re}((-L_{ji}^{\tau} - \overline{L}_{ij}^{\tau} + L_{ii}^{\tau} + \overline{L}_{jj}^{\tau})L_{ij}) \\
& \hspace{1cm} -2(\mathrm{Re}(L_{ii}^{\tau} + L_{jj}^{\tau})|Lij|^2)\big) - \frac{4\kappa_1}{N}\mathrm{Im}L_{ij}\mathrm{Im}(L_{ii}^{\tau} - L_{jj}^{\tau}) \\
& \hspace{1cm} \leq |L_{ij}|\frac{16\mathcal{C}_2|\kappa_1|\tau+8\mathcal{C}_2\kappa_0\tau}{N} + |L_{ij}|^2\frac{4\kappa_0 \mathcal{C}_2\tau}{N} .
\end{aligned}
\end{align*}
\\
$\bullet$ Step C:  Finally, we analyze the term
\[
\frac{2\kappa_0}{N}\sum_{k=1}^{N}|L_{ij}|^2(\mathrm{Re}L_{ki}^{\tau} + \mathrm{Re}L_{kj}^{\tau} - 2) + \frac{4\kappa_1}{N}\sum_{k = 1}^{N}\mathrm{Im}L_{ij}(\mathrm{Im}L_{ki}^{\tau} - \mathrm{Im}L_{kj}^{\tau}).
\] 
By direct calculation, we have
\begin{align*}
&\frac{2\kappa_0}{N}\sum_{k=1}^{N}|L_{ij}|^2(\mathrm{Re}L_{ki}^{\tau} + \mathrm{Re}L_{kj}^{\tau} - 2) + \frac{4\kappa_1}{N}\sum_{k = 1}^{N}\mathrm{Im}L_{ij}(\mathrm{Im}L_{ki}^{\tau} - \mathrm{Im}L_{kj}^{\tau})\\
& \hspace{0.5cm} = \frac{2\kappa_0}{N}\sum_{k=1}^{N}|1-\langle z_i, z_j \rangle|^2(\mathrm{Re}(1- \langle z_k^{\tau}, z_i \rangle + 1 - \langle z_k^{\tau}, z_j \rangle - 2) \\
&\hspace{0.7cm} + \frac{4\kappa_1}{N}\sum_{k=1}^{N}\mathrm{Im}(1- \langle z_i, z_j \rangle)(\mathrm{Im}(1- \langle z_k^{\tau}, z_i \rangle - 1 + \langle z_k^{\tau}, z_j \rangle)  \\
& \hspace{0.5cm} = -\frac{2\kappa_0}{N}\sum_{k=1}^{N}|1-\langle z_i, z_j \rangle|^2(\mathrm{Re}( \langle z_k^{\tau}, z_i + z_j \rangle) + \frac{4\kappa_1}{N}\sum_{k=1}^{N}\mathrm{Im}(\langle z_i, z_j \rangle)(\mathrm{Im}( \langle z_k^{\tau}, z_i - z_j \rangle).  \\
\end{align*}
Note that 
\begin{align*}
\begin{aligned}
&\frac{4\kappa_1}{N}\sum_{k=1}^{N}\mathrm{Im}(\langle z_i, z_j \rangle)(\mathrm{Im}( \langle z_k^{\tau}, z_i - z_j \rangle) =  4\kappa_1 \mathrm{Im}(\langle z_i, z_j \rangle)(\mathrm{Im}( \langle z_c^{\tau}, z_i - z_j \rangle)\\
&\hspace{2cm}\leq 4|\kappa_1|\cdot|L_{ij}|\cdot(|L_{ci}^{\tau}| + |L_{cj}^{\tau}|) \leq 4|\kappa_1| \cdot |L_{ij}|\cdot(|L_{ci}|+|L_{cj}|+ 2\mathcal{C}_2\tau),
\end{aligned}
\end{align*}
and so
\begin{align*}
&-\frac{2\kappa_0}{N}\sum_{k=1}^{N}|1-\langle z_i, z_j \rangle|^2(\mathrm{Re}( \langle z_k^{\tau}, z_i + z_j \rangle) + \frac{4\kappa_1}{N}\sum_{k=1}^{N}\mathrm{Im}(\langle z_i, z_j \rangle)(\mathrm{Im}( \langle z_k^{\tau}, z_i - z_j \rangle) \\
&\hspace{1cm} \leq -\frac{2\kappa_0}{N}\sum_{k=1}^{N}|1-\langle z_i, z_j \rangle|^2(\mathrm{Re}( \langle z_k^{\tau}, z_i + z_j \rangle) + 4|\kappa_1|\cdot |L_{ij}|\cdot(|L_{ci}|+|L_{cj}|+ 2\mathcal{C}_2\tau).
\end{align*}

\vspace{0.2cm}

$\bullet$ Step D: We collect all the estimates in Step A - Step C to find 
\begin{align*}
\frac{d}{dt}|L_{ij}|^2 &= \frac{2\kappa_0}{N}\sum_{k=1}^{N}|L_{ij}|^2(\mathrm{Re}L_{ki}^{\tau} + \mathrm{Re}L_{kj}^{\tau} - 2) + \frac{4\kappa_1}{N}\sum_{k = 1}^{N}\mathrm{Im}L_{ij}(\mathrm{Im}L_{ki}^{\tau} - \mathrm{Im}L_{kj}^{\tau})\\
&\hspace{0.5cm} + \frac{\kappa_0}{N}\big(4|L_{ij}|^2 + 2\mathrm{Re}((-L_{ji}^{\tau} - \overline{L}_{ij}^{\tau} + L_{ii}^{\tau} + \overline{L}_{jj}^{\tau})L_{ij}) -2(\mathrm{Re}(L_{ii}^{\tau} + L_{jj}^{\tau})|L_{ij}|^2) \\
&\hspace{0.5cm} - \frac{4\kappa_1}{N}\mathrm{Im}L_{ij}\mathrm{Im}(L_{ii}^{\tau} - L_{jj}^{\tau}) \\
&\leq -\frac{2\kappa_0}{N}\sum_{k=1}^{N}|L_{ij}|^2(\mathrm{Re}( \langle z_k^{\tau}, z_i + z_j \rangle) + 4|\kappa_1|\cdot |L_{ij}|\cdot(|L_{ci}|+|L_{cj}|+ 2\mathcal{C}_2\tau) \\
&\hspace{0.5cm}+|L_{ij}|\frac{16\mathcal{C}_2|\kappa_1|\tau+8\mathcal{C}_2\kappa_0\tau}{N}
 + |L_{ij}|^2\frac{4\kappa_0 \mathcal{C}_2\tau}{N}.
\end{align*}
\end{proof}
We set  
\[ L(t) = \max_{ i,j}|L_{ij}|. \]
Then, for each time $t$, there exists $i_t$ and $j_t$ by which the maximum is attained, i.e.
\begin{align}\label{D-6}
L(t) = |1-\langle z_{i_t}, z_{j_t} \rangle|.
\end{align}
Now we want to obtain the dynamics of $L(t)$.

\begin{lemma} \label{L4.5}
Let $\{z_j\}$ be a global solution to \eqref{C-8}. Then, the functional $L(t)$ in \eqref{D-6} satisfies 
\begin{align*}
\begin{aligned}
\frac{d}{dt}L(t) &\leq 2\kappa_0 L(t)^2 + \left(-2\kappa_0 +2\mathcal{C}_2\kappa_0 \tau +4|\kappa_1| + \frac{2\mathcal{C}_2\kappa_0 \tau}{N}\right)L(t) \\
&+ \left(4\mathcal{C}_2|\kappa_1|\tau + \frac{4\mathcal{C}_2\tau}{N}(\kappa_0 + 2|\kappa_1|)\right).
\end{aligned}
\end{align*}
\end{lemma}
\begin{proof}
It follows from Lemma \ref{L4.4} that 
\begin{align*}
\frac{d}{dt}L(t)^2  &\leq -2\kappa_0 L(t)^2(\mathrm{Re}( \langle z_c^{\tau}, z_i + z_j \rangle) + 4|\kappa_1|\cdot L(t)\cdot(2L(t)+ 2\mathcal{C}_2\tau)\\
 &\hspace{0.5cm}  +L(t)\frac{16\mathcal{C}_2|\kappa_1|\tau+8\mathcal{C}_2\kappa_0\tau}{N}+ L(t)^2\frac{4\kappa_0 \mathcal{C}_2\tau}{N} \\
&=  \left(-2\kappa_0\mathrm{Re}( \langle z_c^{\tau}, z_i + z_j \rangle) +8|\kappa_1| + \frac{4\mathcal{C}_2\kappa_0 \tau}{N}\right)L(t)^2 \\
& \hspace{0.5cm} + \left(8\mathcal{C}_2|\kappa_1|\tau + \frac{16\mathcal{C}_2|\kappa_1|\tau+8\mathcal{C}_2\kappa_0\tau}{N}\right)L(t).
\end{align*}
Note that 
\[
\mathrm{Re}( \langle z_c^{\tau}, z_i + z_j \rangle) = - \mathrm{Re}(L_{ci}^{\tau} + L_{cj}^{\tau}) + 2. 
\]
Thus, we have
\begin{align*}
\begin{aligned}
\frac{d}{dt}L(t)^2 &\leq \Big(-2\kappa_0\mathrm{Re}(2-L_{ci}^{\tau} -L_{cj}^{\tau} ) +8|\kappa_1| + \frac{4\mathcal{C}_2\kappa_0 \tau}{N} \Big)L(t)^2  \\
& \hspace{0.5cm} + \left(8\mathcal{C}_2|\kappa_1|\tau + \frac{16\mathcal{C}_2|\kappa_1|\tau+8\mathcal{C}_2\kappa_0\tau}{N}\right)L(t) \\
&\leq 4\kappa_0 L(t)^3 + \left(-4\kappa_0 +4\mathcal{C}_2\kappa_0 \tau +8|\kappa_1| + \frac{4\mathcal{C}_2\kappa_0 \tau}{N}\right)L(t)^2  \\
& \hspace{0.5cm} + \left(8\mathcal{C}_2|\kappa_1|\tau + \frac{16\mathcal{C}_2|\kappa_1|\tau+8\mathcal{C}_2\kappa_0\tau}{N}\right)L(t).
\end{aligned}
\end{align*}
Hence we obtain the desired estimate:
\begin{align*}
\begin{aligned}
\frac{d}{dt}L(t) &\leq 2\kappa_0 L(t)^2 + \left(-2\kappa_0 +2\mathcal{C}_2\kappa_0 \tau +4|\kappa_1| + \frac{2\mathcal{C}_2\kappa_0 \tau}{N}\right)L(t)  \\
&+ \left(4\mathcal{C}_2|\kappa_1|\tau + \frac{4\mathcal{C}_2\tau}{N}(\kappa_0 + 2|\kappa_1|)\right).
\end{aligned}
\end{align*}
\end{proof} 

\subsubsection{Proof of Theorem \ref{T4.1}} \label{sec:4.1.2}
Consider a polynomial 
\begin{align*}
f(x)&=  2\kappa_0 x^2 +  \left(-2\kappa_0 +2\mathcal{C}_2\kappa_0 \tau +4|\kappa_1| + \frac{2\mathcal{C}_2\kappa_0 \tau}{N}\right)x +  \left(4\mathcal{C}_2|\kappa_1|\tau + \frac{4\mathcal{C}_2\tau}{N}(\kappa_0 + 2|\kappa_1|)\right)\\
&=2\kappa_0\left(x^2-\left(1-\frac{2|\kappa_1|}{\kappa_0}-\mathcal{C}_2\tau\left(\frac{N+1}{N}\right)x\right)+\mathcal{C}_2\tau\left(\frac{2|\kappa_1|}{\kappa_0}\left(1+\frac{2}{N}\right)+\frac{1}{N}\right)\right).
\end{align*}
Now we study about the practical synchronization with $\tau\searrow0$. Here we fix the other variables. Let assume that $2|\kappa_1|<\kappa_0$. Then for a sufficiently small $\tau$, there are two solutions of $f(x)=0$. If they exist, let they be $x_-(\tau)$ and $x_+(\tau)$ with $x_-(\tau)\leq x_+(\tau)$. Then we have following property from the phase portrait:
\[
L(0)<x_+(\tau)\quad\Rightarrow\quad\limsup_{t\to\infty}L(t)\leq x_-(\tau).
\]
Also we can obtain
\[
\lim_{\tau\searrow0}x_-(\tau)=0,\quad \lim_{\tau\searrow0}x_+(\tau)=1-\frac{2|\kappa_1|}{\kappa_0}.
\]

\subsection{General network topology} \label{sec:4.2} In this subsection, we will study practical synchronization of the system \eqref{A-1}. Here we set network topology as $\{a_{ij}\}$ with $a_{ij} \geq 0$, for all $i,j \in \mathcal{N}$. 
\begin{align}\label{D-7}
\begin{cases}
\dot{z}_j=\Omega_j z_j + \displaystyle\frac{\kappa_0}{N}\sum_{k \neq j}a_{jk}\big(\langle z_j, z_j \rangle z_k^{\tau} - \langle z_k^{\tau}, z_j \rangle z_j) + \frac{\kappa_1}{N}\sum_{k \neq j}a_{jk}\big(\langle z_j, z_k^{\tau} \rangle - \langle z_k^{\tau} , z_j \rangle) z_j, \\
z_j(t) = \varphi_j(t)\in\bbc^d, \quad -\tau \leq t \leq 0 ,
\end{cases}
\end{align}
\begin{theorem} \label{T4.2}
Let $\{z_j\}$ be a global solution to \eqref{D-7} with the following initial condition:
\[
L(0)<1-\frac{2\sum_{k=1}^{N}|a_{ik} - a_{jk}|}{\sum_{k=1}^N(a_{ik} + a_{jk})}.
\]
 Then, system \eqref{D-7} exhibits the practical synchronization:
\[
\lim_{\kappa_0\to\infty}\lim_{\tau\searrow0}\limsup_{t\to\infty}L(t)=0.
\]
\end{theorem}}
\begin{proof}
We leave its proof in Section \ref{sec:4.2.2}
\end{proof}

\subsubsection{Basic estimates} \label{sec:4.2.1} In this part, we provide several a priori estimates to be used in the proof of Theorem \ref{T4.2}. 
\begin{lemma} \label{L4.6}
Let $\{z_j\}$ be a global solution to \eqref{D-7}. Then we have 
\begin{align*}
\frac{d}{dt} G_{ij} &=\langle (\Omega_i-\Omega_j) z_i, z_j \rangle + \displaystyle\frac{\kappa_0}{N}\sum_{k \neq i}a_{ik}(G_{kj}^{\tau}-\overline{G}_{ki}^{\tau}G_{ij}) + \displaystyle\frac{\kappa_0}{N}\sum_{k \neq j}a_{jk}(\overline{G}_{ki}^{\tau} - G_{kj}^{\tau}G_{ij})
\\
&\hspace{0.5cm} + \displaystyle\frac{\kappa_1}{N}\sum_{k \neq i}a_{ik}(G_{ki}^{\tau}-\overline{G}_{ki}^{\tau})G_{ij} + \displaystyle\frac{\kappa_1}{N}\sum_{k \neq j}a_{jk}(\overline{G}_{kj}^{\tau} - G_{kj}^{\tau})G_{ij}.
\end{align*}
\end{lemma}
\begin{proof}
By definition, one has 
\begin{align*}
\frac{d}{dt} G_{ij} &=\langle \dot{z}_i, z_j\rangle+\langle z_i, \dot{z}_j\rangle\\
&= \left\langle \Omega_i z_i + \displaystyle\frac{\kappa_0}{N}\sum_{k \neq i}a_{ik}\big(z_k^{\tau} - \langle z_k^{\tau}, z_i \rangle z_i) + \frac{\kappa_1}{N}\sum_{k \neq i}a_{ik}\big(\langle z_i, z_k^{\tau} \rangle - \langle z_k^{\tau} , z_i \rangle) z_i, z_j \right\rangle\\
&\hspace{0.5cm} + \left\langle z_i, \Omega_j z_j + \displaystyle\frac{\kappa_0}{N}\sum_{k \neq j}a_{jk}\big(z_k^{\tau} - \langle z_k^{\tau}, z_j \rangle z_j) + \frac{\kappa_1}{N}\sum_{k \neq j}a_{jk}\big(\langle z_j, z_k^{\tau} \rangle - \langle z_k^{\tau} , z_j \rangle) z_j \right\rangle\\
&= \langle \Omega_i z_i, z_j \rangle + \displaystyle\frac{\kappa_0}{N}\sum_{k \neq i}a_{ik}\big(\langle z_k^{\tau}, z_j \rangle - \overline{\langle z_k^{\tau}, z_i \rangle} \langle z_i, z_j \rangle) + \frac{\kappa_1}{N}\sum_{k \neq i}a_{ik}\big(\overline{\langle z_i, z_k^{\tau} \rangle} - \overline{\langle z_k^{\tau} , z_i \rangle}) \langle z_i, z_j \rangle\\
&\hspace{0.5cm} + \langle z_i, \Omega_j z_j \rangle + \displaystyle\frac{\kappa_0}{N}\sum_{k \neq j}a_{jk}\big( \langle z_i, z_k^{\tau} \rangle - \langle z_k^{\tau}, z_j \rangle \langle z_i, z_j \rangle) + \frac{\kappa_1}{N}\sum_{k \neq j}a_{jk}\big(\langle z_j, z_k^{\tau} \rangle  - \langle z_k^{\tau} , z_j \rangle) \langle z_i, z_j \rangle \\
&= \langle (\Omega_i-\Omega_j) z_i, z_j \rangle + \displaystyle\frac{\kappa_0}{N}\sum_{k \neq i}a_{ik}\big(\langle z_k^{\tau}, z_j \rangle - \overline{\langle z_k^{\tau}, z_i \rangle} \langle z_i, z_j \rangle) + \frac{\kappa_1}{N}\sum_{k \neq i}a_{ik}\big(\overline{\langle z_i, z_k^{\tau} \rangle} - \overline{\langle z_k^{\tau} , z_i \rangle}) \langle z_i, z_j \rangle\\
&\hspace{0.5cm} + \displaystyle\frac{\kappa_0}{N}\sum_{k \neq j}a_{jk}\big( \langle z_i, z_k^{\tau} \rangle - \langle z_k^{\tau}, z_j \rangle \langle z_i, z_j \rangle) + \frac{\kappa_1}{N}\sum_{k \neq j}a_{jk}\big(\langle z_j, z_k^{\tau} \rangle  - \langle z_k^{\tau} , z_j \rangle) \langle z_i, z_j \rangle \\
&=\langle (\Omega_i-\Omega_j) z_i, z_j \rangle + \displaystyle\frac{\kappa_0}{N}\sum_{k \neq i}a_{ik}(G_{kj}^{\tau}-\overline{G}_{ki}^{\tau}G_{ij}) + \displaystyle\frac{\kappa_0}{N}\sum_{k \neq j}a_{jk}(\overline{G}_{ki}^{\tau} - G_{kj}^{\tau}G_{ij})
\\
&\hspace{0.5cm} + \displaystyle\frac{\kappa_1}{N}\sum_{k \neq i}a_{ik}(G_{ki}^{\tau}-\overline{G}_{ki}^{\tau})G_{ij} + \displaystyle\frac{\kappa_1}{N}\sum_{k \neq j}a_{jk}(\overline{G}_{kj}^{\tau} - G_{kj}^{\tau})G_{ij}.
\end{align*}
\end{proof}
\begin{lemma} \label{L4.7}
Let $\{z_j\}$ be a global solution to \eqref{D-7}. Then, one has 
\[ 
| L_{ij}(t)-L_{ij}^{\tau}(t)| \leq \tau \mathcal{C}_3,
\]
where the positive constant $\mathcal{C}_3$ is given by
\[
\mathcal{C}_3 = \left(\max_{i}\|\Omega_{i}\|_{\infty} + \frac{2\mathcal{D}(A)(N-1)(\kappa_0+|\kappa_1|)}{N} \right).
\]
\end{lemma}
\begin{proof}
By the Cauchy-Schwarz inequality, we have
\[
|L_{ij}(t)-L_{ij}^{\tau}(t)| = |\langle z_i - z_i^{\tau} , z_j \rangle| \leq \|z_i - z_i^{\tau} \| \cdot\|z_j \| =\|z_i - z_i^{\tau} \|.
\]
On the other hand, we have
\begin{align*}
&\|z_i(t) - z_i^\tau(t)\| = \left\|\int_{t- \tau}^{t}\dot{z}_i(s)ds\right\| \\
& \hspace{0.5cm}  \leq \int_{t- \tau}^{t}\left\|\Omega_iz_i + \displaystyle\frac{\kappa_0}{N}\sum_{k \neq i}a_{ik}\big(\langle z_i, z_i \rangle z_k^{\tau} - \langle z_k^{\tau}, z_i \rangle z_i) + \frac{\kappa_1}{N}\sum_{k \neq i}a_{ik}\big(\langle z_i, z_k^{\tau} \rangle - \langle z_k^{\tau} , z_i \rangle) z_i\right\|ds \\
& \hspace{0.5cm} \leq \tau \left(\|\Omega_i\|_{\infty} + \frac{2a_{ik}(N-1)(\kappa_0+|\kappa_1|)}{N} \right).
\end{align*}
Now, we set
\[  \|A \|_{\infty}:= \displaystyle\max_{i,j}{a_{ij}}. \]
Then, we have 
\[
|L_{ij}(t)-L_{ij}^{\tau}(t)| \leq \tau \left(\max_{i}\|\Omega_{i}\|_{\infty} + \frac{2 \|A \|_{\infty} (N-1)(\kappa_0+|\kappa_1|)}{N} \right).
\]
\end{proof}
\noindent We set 
\[  \mathcal{D}(\Omega) := \displaystyle\max_{i,j}{\|\Omega_i - \Omega_j\|_{\infty}}. \]
Then by direct calculation, we get the following lemma. 
\begin{lemma} \label{L4.8}
Let $\{z_j\}$ be a global solution to \eqref{D-7}. Then $|L_{ij}|$ satisfies
\begin{align}
\begin{aligned} \label{D-8}
\frac{d}{dt}|L_{ij}|^2 &\leq 2|L_{ij}|\mathcal{D}(\Omega) + \frac{2\kappa_0}{N}\sum_{k=1}^{N}|a_{ik} - a_{jk}|(|L_{ki}| + |L_{kj}|+2\tau \mathcal{C}_3)|L_{ij}| + \frac{4\mathcal{C}_3|\kappa_1|\tau}{N}(a_{ii}+a_{jj})|L_{ij}|\\
&\hspace{0.5cm} -2\frac{\kappa_0}{N}\sum_{k=1}^{N}(a_{ik}+a_{jk})|L_{ij}|^2 + \frac{2\kappa_0}{N}\sum_{k=1}^{N} \bigg(a_{ik}(|L_{ki}|+\tau \mathcal{C}_3)+a_{jk}(|L_{kj}|+\tau \mathcal{C}_3 \bigg)|L_{ij}|^2 \\
&\hspace{0.5cm}+\frac{2C\kappa_0 \tau}{N}(a_{ii} + a_{jj})|L_{ij}| + \frac{2\mathcal{C}_3\kappa_0 \tau}{N}(a_{ii} + a_{jj})(|L_{ij}|^2 + |L_{ij}|)\\
&\hspace{0.5cm}+\frac{4|\kappa_1|}{N}\sum_{k=1}^{N}|L_{ij}|(a_{ik}(|L_{ki}| + \tau \mathcal{C}_3) + a_{jk}(|L_{kj}| + \tau \mathcal{C}_3)). 
\end{aligned}
\end{align}
\end{lemma}
\begin{proof} We use \eqref{D-7} to find 
\begin{align*}
\frac{d}{dt}\langle z_i, z_j\rangle &=\langle (\Omega_i-\Omega_j) z_i, z_j \rangle + \displaystyle\frac{\kappa_0}{N}\sum_{k \neq i}a_{ik}(G_{kj}^{\tau}-\overline{G}_{ki}^{\tau}G_{ij}) + \displaystyle\frac{\kappa_0}{N}\sum_{k \neq j}a_{jk}(\overline{G}_{ki}^{\tau} - G_{kj}^{\tau}G_{ij})\\
&\hspace{0.5cm} + \displaystyle\frac{\kappa_1}{N}\sum_{k \neq i}a_{ik}(G_{ki}^{\tau}-\overline{G}_{ki}^{\tau})G_{ij} + \displaystyle\frac{\kappa_1}{N}\sum_{k \neq j}a_{jk}(\overline{G}_{kj}^{\tau} - G_{kj}^{\tau})G_{ij}\\ 
&= \langle (\Omega_i-\Omega_j) z_i, z_j \rangle + \displaystyle\frac{\kappa_0}{N}\sum_{k = 1}^{N}\left(a_{ik}(G_{kj}^{\tau}-\overline{G}_{ki}^{\tau}G_{ij}) + a_{jk}(\overline{G}_{ki}^{\tau} - G_{kj}^{\tau}G_{ij})\right)\\
&\hspace{0.5cm}+ \displaystyle\frac{\kappa_1}{N}\sum_{k =1}^{N}\left(a_{ik}(G_{ki}^{\tau}G_{ij}-\overline{G}_{ki}^{\tau}G_{ij})+a_{jk}(\overline{G}_{kj}^{\tau}G_{ij} - G_{kj}^{\tau}G_{ij})\right) \\
&\hspace{0.5cm}- \frac{\kappa_0}{N}\left(a_{ii}(G_{ij}^{\tau}-\overline{G}_{ii}^{\tau}G_{ij}) + a_{jj}(\overline{G}_{ji}^{\tau}-G_{jj}^{\tau}G_{ij})\right)\\
&\hspace{0.5cm}- \frac{\kappa_1}{N}\bigg(a_{ii}(G_{ii}^{\tau}G_{ij}-\overline{G}_{ii}^{\tau}G_{ij}) + a_{jj}(\overline{G}_{jj}^{\tau}G_{ij}-G_{jj}^{\tau}G_{ij})\bigg)\\
&= \langle (\Omega_i-\Omega_j) z_i, z_j \rangle\\
&\hspace{0.5cm} + \frac{\kappa_0}{N}\sum_{k=1}^{N}\left[(a_{jk}-a_{ik})L_{kj}^{\tau} + (a_{ik} - a_{jk})\overline{L}_{ki}^{\tau} + (a_{ik}+a_{jk} - a_{ik}\overline{L}_{ki}^{\tau} - a_{jk}L_{kj}^\tau)L_{ij}\right]\\
&\hspace{0.5cm}+\frac{2i\kappa_1}{N}\sum_{k=1}^{N}\left(a_{jk}\mathrm{Im}L_{kj}^{\tau}- a_{ik}\mathrm{Im}L_{ki}^{\tau})(1-L_{ij})\right) - \frac{2\mathrm{i}\kappa_1}{N}(a_{jj}\mathrm{Im}L_{jj}^{\tau} - a_{ii}\mathrm{Im}L_{ii}^\tau)(1-L_{ij}) \\
&\hspace{0.5cm}-\frac{\kappa_0}{N}\left(a_{jj}L_{jj}^\tau - a_{ii}L_{ij}^\tau + a_{ii}\overline{L}_{ii}^\tau-a_{jj}\overline{L}_{ji}^\tau + (a_{ii} + a_{jj})L_{ij} -a_{ii}\overline{L}_{ii}^\tau L_{ij} - a_{jj}L_{jj}^\tau L_{ij}\right).
\end{align*}
Thus, we have
\begin{align*}
\frac{d}{dt}|L_{ij}|^2 &= \frac{d}{dt}(L_{ij} \overline{L}_{ij}) = 2\re(\dot{L_{ij}}\overline{L}_{ij}) = -2\re\left(\frac{d}{dt}\langle z_i, z_j \rangle (1 - \langle z_j, z_i \rangle)\right)\\
&=-2\re \Big(L_{ji}\bigg[\langle (\Omega_i-\Omega_j) z_i, z_j \rangle + \frac{\kappa_0}{N}\sum_{k=1}^{N}\big\{(a_{jk}-a_{ik})L_{kj}^{\tau}\\
&\hspace{0.5cm}+ (a_{ik} - a_{jk})\overline{L}_{ki}^{\tau} + (a_{ik}+a_{jk} - a_{ik}\overline{L}_{ki}^{\tau} - a_{jk}L_{kj}^\tau)L_{ij}\big\}\\
&\hspace{0.5cm}+\frac{2\mathrm{i}\kappa_1}{N}\sum_{k=1}^{N}\left(a_{jk}\mathrm{Im}L_{kj}^{\tau}- a_{ik}\mathrm{Im}L_{ki}^{\tau})(1-L_{ij})\right) - \frac{2\mathrm{i}\kappa_1}{N}(a_{jj}\mathrm{Im}L_{jj}^{\tau} - a_{ii}\mathrm{Im}L_{ii}^\tau)(1-L_{ij}) \\
&\hspace{0.5cm}-\frac{\kappa_0}{N}\left(a_{jj}L_{jj}^\tau - a_{ii}L_{ij}^\tau + a_{ii}\overline{L}_{ii}^\tau-a_{jj}\overline{L}_{ji}^\tau + (a_{ii} + a_{jj})L_{ij} -a_{ii}\overline{L}_{ii}^\tau L_{ij} - a_{jj}L_{jj}^\tau L_{ij}\right)\bigg]\Big) \\
&\leq 2|L_{ij}|\mathcal{D}(\Omega) -2\re \frac{\kappa_0}{N}\sum_{k=1}^{N}\left\{(a_{ik}+a_{jk} - a_{ik}\overline{L}_{ki}^{\tau} - a_{jk}L_{kj}^\tau)|L_{ij}|^2 + (a_{ik} - a_{jk})(\overline{L}_{ki}^{\tau}-L_{kj}^\tau)L_{ji}\right\} \\
&\hspace{0.5cm}+\re\frac{2\kappa_0}{N}((a_{ii} + a_{jj} - a_{ii}\overline{L}_{ii}^\tau - a_{jj}L_{jj}^\tau)|L_{ij}|^2 + (a_{jj}L_{jj}^\tau - a_{ii}L_{ij}^\tau + a_{ii}\overline{L}_{ii}^\tau-a_{jj}\overline{L}_{ji}^\tau)L_{ji})\\
&\hspace{0.5cm}+\frac{4\kappa_1}{N}\sum_{k=1}^{N}(\mathrm{Im}(a_{ik}L_{ki}^\tau - a_{jk}L_{kj}^\tau)\mathrm{Im}L_{ij} - \frac{4\kappa_1}{N}\mathrm{Im}(a_{ii}L_{ii}^\tau - a_{jj}L_{jj}^\tau)\mathrm{Im}L_{ij}.\\
\end{align*}
In the sequel, we prove each term in the R.H.S. of the above relation. \newline

\vspace{0.2cm}

\noindent $\bullet$~Step A: By direct calculation, one has 
\begin{align*}
\begin{aligned}
& \left|\frac{4\kappa_1}{N}\mathrm{Im}(a_{ii}L_{ii}^\tau - a_{jj}L_{jj}^\tau)\mathrm{Im}L_{ij} \right| = \frac{4|\kappa_1|}{N}|L_{ij}||\mathrm{Im}(a_{ii}L_{ii}^\tau - a_{jj}L_{jj}^\tau)|\\
& \hspace{0.5cm} \leq \frac{4|\kappa_1|}{N}|L_{ij}|(a_{ii}|L_{ii}^\tau| + a_{jj}|L_{jj}^\tau|)  \leq \frac{4\mathcal{C}_3|\kappa_1|\tau}{N}(a_{ii}+a_{jj})|L_{ij}|
\end{aligned}
\end{align*}

\vspace{0.2cm}

\noindent $\bullet$ Step B:~We set 
\[
A := \frac{2\kappa_0}{N}\re\bigg((a_{ii} + a_{jj} - a_{ii}\overline{L}_{ii}^\tau - a_{jj}L_{jj}^\tau)|L_{ij}|^2 + (a_{jj}L_{jj}^\tau - a_{ii}L_{ij}^\tau + a_{ii}\overline{L}_{ii}^\tau-a_{jj}\overline{L}_{ji}^\tau)L_{ji}\bigg).
\]
Then, one has
\begin{align*}
A &\leq \frac{2\kappa_0}{N}\re \left(a_{ii}(L_{ij} - L_{ij}^\tau)L_{ji} + a_{jj}(L_{ij} - \overline{L}_{ji}^\tau)L_{ji} \right) + \frac{2\mathcal{C}_3\kappa_0 \tau}{N}(a_{ii} + a_{jj})(|L_{ij}|^2 + |L_{ij}|) \\
&\leq \frac{2\mathcal{C}_3\kappa_0 \tau}{N}(a_{ii} + a_{jj})|L_{ij}| + \frac{2\mathcal{C}_3\kappa_0 \tau}{N}(a_{ii} + a_{jj})(|L_{ij}|^2 + |L_{ij}|).
\end{align*}

\vspace{0.2cm}

\noindent $\bullet$ Step C:~We set 
\begin{align*}
B &= -2\re \frac{\kappa_0}{N}\sum_{k=1}^{N}\left\{(a_{ik}+a_{jk} - a_{ik}\overline{L}_{ki}^{\tau} - a_{jk}L_{kj}^\tau)|L_{ij}|^2 + (a_{ik} - a_{jk})(\overline{L}_{ki}^{\tau}-L_{kj}^\tau)L_{ji}\right\} \\
&\leq -2\frac{\kappa_0}{N}\sum_{k=1}^{N}(a_{ik}+a_{jk})|L_{ij}|^2 + \frac{2\kappa_0}{N}\sum_{k=1}^{N} \bigg(a_{ik}(|L_{ki}|+\tau \mathcal{C}_3)+a_{jk}(|L_{kj}|+\tau \mathcal{C}_3 )\bigg)|L_{ij}|^2 \\
&\hspace{0.5cm}+ \frac{2\kappa_0}{N}\sum_{k=1}^{N}|a_{ik} - a_{jk}|(|L_{ki}| + |L_{kj}|+2\tau \mathcal{C}_3)|L_{ij}|.
\end{align*}

\vspace{0.2cm}

\noindent $\bullet$~Step D:~Note that 
\begin{align*}
\frac{4\kappa_1}{N}\sum_{k=1}^{N}(\mathrm{Im}(a_{ik}L_{ki}^\tau - a_{jk}L_{kj}^\tau)\mathrm{Im}L_{ij} \leq \frac{4|\kappa_1|}{N}\sum_{k=1}^{N}|L_{ij}|(a_{ik}(|L_{ki}| + \tau \mathcal{C}_3) + a_{jk}(|L_{kj}| + \tau \mathcal{C}_3)). 
\end{align*}

\vspace{0.2cm}

\noindent $\bullet$ Step E: We collect all the estimates in Step A - Step D to get 
\begin{align*}
\begin{aligned}
\frac{d}{dt}|L_{ij}|^2 &\leq 2|L_{ij}|\mathcal{D}(\Omega) + \frac{2\kappa_0}{N}\sum_{k=1}^{N}|a_{ik} - a_{jk}|(|L_{ki}| + |L_{kj}|+2\tau \mathcal{C}_3)|L_{ij}| + \frac{4\mathcal{C}_3|\kappa_1|\tau}{N}(a_{ii}+a_{jj})|L_{ij}|\\
&-\frac{2\kappa_0}{N}\sum_{k=1}^{N}(a_{ik}+a_{jk})|L_{ij}|^2 + \frac{2\kappa_0}{N}\sum_{k=1}^{N} \bigg(a_{ik}(|L_{ki}|+\tau \mathcal{C}_3)+a_{jk}(|L_{kj}|+\tau \mathcal{C}_3 )\bigg)|L_{ij}|^2 \\
&+\frac{2\mathcal{C}_3\kappa_0 \tau}{N}(a_{ii} + a_{jj})|L_{ij}| + \frac{2\mathcal{C}_3\kappa_0 \tau}{N}(a_{ii} + a_{jj})(|L_{ij}|^2 + |L_{ij}|)\\
&+\frac{4|\kappa_1|}{N}\sum_{k=1}^{N}|L_{ij}|(a_{ik}(|L_{ki}| + \tau \mathcal{C}_3) + a_{jk}(|L_{kj}| + \tau \mathcal{C}_3)).
\end{aligned}
\end{align*}
\end{proof}

\subsubsection{Proof of Theorem \ref{T4.2}} \label{sec:4.2.2} 
Recall that 
\[ L(t) := \max_{i,j}{L_{ij}}. \]
Then, for each time $t$, there exists $i_t$ and $j_t$ by which the maximum is attained:
\[
L(t) = |1-\langle z_{i_t}, z_{j_t} \rangle|.
\]
Then by Lemma \eqref{L4.8}, one has 
\begin{align*}
\frac{d}{dt}L(t)^2  &\leq 2 L(t)\mathcal{D}(\Omega) + \frac{4\kappa_0}{N}\sum_{k=1}^{N}|a_{ik} - a_{jk}|(L(t)+ \tau \mathcal{C}_3)L(t) + \frac{4\mathcal{C}_3|\kappa_1|\tau}{N}(a_{ii}+a_{jj})L(t)\\
& \hspace{0.5cm}-\frac{2\kappa_0}{N}\sum_{k=1}^{N}(a_{ik}+a_{jk})L(t)^2 + \frac{2\kappa_0}{N}\sum_{k=1}^{N} \bigg(a_{ik}(L(t)+\tau \mathcal{C}_3)+a_{jk}(L(t)+\tau \mathcal{C}_3 )\bigg)L(t)^2 \\
&\hspace{0.5cm}+ \frac{2\mathcal{C}_3\kappa_0 \tau}{N}(a_{ii} + a_{jj})L(t) + \frac{2\mathcal{C}_3\kappa_0 \tau}{N}(a_{ii} + a_{jj})(L(t)^2 + L(t))\\
&\hspace{0.5cm}+\frac{4|\kappa_1|}{N}\sum_{k=1}^{N}L(t)\bigg(a_{ik}(L(t) + \tau \mathcal{C}_3) + a_{jk}(L(t) + \tau \mathcal{C}_3)\bigg) \\
&\leq \frac{2\kappa_0}{N}\sum_{k=1}^N(a_{ik} + a_{jk})L(t)^3 \\
&\hspace{0.5cm}+ \bigg( -\frac{2\kappa_0}{N}\sum_{k=1}^{N}(a_{ik}+a_{jk}) + \frac{4\kappa_0}{N}\sum_{k=1}^{N}|a_{ik} - a_{jk}|+ \frac{4\mathcal{C}_3\kappa_0 \tau}{N}\sum_{k=1}^{N}(a_{ik} + a_{jk})\\
&\hspace{0.5cm}+ \frac{2\mathcal{C}_3\kappa_0 \tau}{N}(a_{ii} + a_{jj})+\frac{4|\kappa_1|}{N}\sum_{k=1}^{N}(a_{ik}+ a_{jk})\bigg)L(t)^2\\
&\hspace{0.5cm} + \bigg(2\mathcal{D}(\Omega) + \frac{4\mathcal{C}_3\kappa_0 \tau}{N}\sum_{k=1}^{N}|a_{ik} - a_{jk}| + (a_{ii} + a_{jj})\frac{4\mathcal{C}_3\tau}{N}(\kappa_0+|\kappa_1|)+8\mathcal{C}_3|\kappa_1|\tau\bigg)L(t).
\end{align*}
Then, we have
\begin{align*}
\frac{d}{dt}L(t)&\leq \frac{\kappa_0}{N}\sum_{k=1}^N(a_{ik} + a_{jk})L(t)^2 \\
&\hspace{0.5cm}+ \bigg( -\frac{\kappa_0}{N}\sum_{k=1}^{N}(a_{ik}+a_{jk}) + \frac{2\kappa_0}{N}\sum_{k=1}^{N}|a_{ik} - a_{jk}|+ \frac{2\mathcal{C}_3\kappa_0 \tau}{N}\sum_{k=1}^{N}(a_{ik} + a_{jk})\\
&\hspace{5cm}+ \frac{\mathcal{C}_3\kappa_0 \tau}{N}(a_{ii} + a_{jj})+\frac{2|\kappa_1|}{N}\sum_{k=1}^{N}(a_{ik}+ a_{jk})\bigg)L(t)\\
&\hspace{0.5cm} + \bigg(\mathcal{D}(\Omega) + \frac{2\mathcal{C}_3\kappa_0 \tau}{N}\sum_{k=1}^{N}|a_{ik} - a_{jk}| + \frac{2\mathcal{C}_3\tau}{N}(a_{ii} + a_{jj})(\kappa_0+|\kappa_1|)+4\mathcal{C}_3|\kappa_1|\tau\bigg).
\end{align*}
Now, we set 
\begin{align*}
\mathcal{A}_1& := \frac{1}{N}\sum_{k=1}^N(a_{ik} + a_{jk}),\\
\mathcal{A}_2& :=\frac{2}{N}\sum_{k=1}^{N}|a_{ik} - a_{jk}|+ \frac{2\mathcal{C}_3 \tau}{N}\sum_{k=1}^{N}(a_{ik} + a_{jk})+ \frac{\mathcal{C}_3\tau}{N}(a_{ii} + a_{jj})+\frac{2|\kappa_1|}{N\kappa_0}\sum_{k=1}^{N}(a_{ik}+ a_{jk}),\\
\mathcal{A}_3& :=\frac{\mathcal{D}(\Omega)}{\kappa_0} + \frac{2\mathcal{C}_3 \tau}{N}\sum_{k=1}^{N}|a_{ik} - a_{jk}| +\frac{2\mathcal{C}_3\tau}{N} (a_{ii} + a_{jj})\left(1+\frac{|\kappa_1|}{\kappa_0}\right)+4\mathcal{C}_3\tau\frac{|\kappa_1|}{\kappa_0}.
\end{align*}
This yields
\[
\frac{d}{dt}L\leq \kappa_0\big(\mathcal{A}_1 L^2-(\mathcal{A}_1-\mathcal{A}_2)L+\mathcal{A}_3\big).
\]
If we impose following conditions:
\begin{align}\label{S-2}
\tau\searrow0\quad\text{and after that}\quad \kappa_0\to\infty,
\end{align}
we obtain
\[
\lim_{\kappa_0\to\infty}\lim_{\tau\searrow0}\mathcal{A}_2= \frac{2}{N}\sum_{k=1}^{N}|a_{ik} - a_{jk}|,\quad \lim_{\kappa_0\to\infty}\lim_{\tau\searrow0}\mathcal{A}_3= 0.
\]
Since the zeros of 
\[
\mathcal{A}_1 x^2-(\mathcal{A}_1-\mathcal{A}_2)x+\mathcal{A}_3=0
\]
can be expressed as
\[
x_\pm=\frac{\mathcal{A}_1-\mathcal{A}_2\pm\sqrt{(\mathcal{A}_1-\mathcal{A}_2)^2-4\mathcal{A}_1\mathcal{A}_3}}{2\mathcal{A}_1}.
\]
If we combine this expression, under the condition \eqref{S-2} we have
\[
\lim_{\kappa_0\to\infty}\lim_{\tau\searrow0}x_+=1-\frac{2\sum_{k=1}^{N}|a_{ik} - a_{jk}|}{\sum_{k=1}^N(a_{ik} + a_{jk})},\quad \lim_{\kappa_0\to\infty}\lim_{\tau\searrow0}x_-=0.
\]
By similar arguments with previous result, we have desired estimate.
\qed

\section{Conclusion} \label{sec:5}
In this paper, we have proposed several sufficient frameworks leading to complete aggregation and practical aggregation in terms of initial data, coupling gains and size of time-delay for the LHS model with time-delayed interactions. The LHS model is a complex counterpart of the LS model, and it describes the continuous-time dynamics of the Lohe hermitian sphere particles on the unit hermitian sphere. For the SL coupling gain pair with $\kappa_1 = -\frac{\kappa_0}{2}$, the LHS model on $\bbc^d$ can be rewritten as the LS model on $\bbr^{2d}$. When the coupling gain pair is close to that of the SL coupling gain pair and the corresponding linear flows are the same, we show that the LHS flow with a time-delay tends to complete aggregation asymptotically for some admissible class of initial data and system parameters. For a general network, we also provided a sufficient framework for practical aggregation to the LHS model with respect to time-delay. The current work only dealt with the same free flow, and we did not take into account of heterogeneous free flow. Even for the LS model on the unit sphere in Euclidean space, the complete aggregation is not known in previous literature, except a weak aggregation estimates such as practical aggregation. We leave this interesting issue for a future work.

\end{document}